%% file: IEEE_TCOM.tex
\documentclass[journal,draftclsnofoot,onecolumn,12pt,twoside]{IEEEtranTCOM}
\usepackage{amsmath,amsfonts}
\usepackage{algorithmic}
\usepackage{algorithm}
\usepackage{array}
\usepackage{subfigure}
\usepackage{textcomp}
\usepackage{stfloats}
\usepackage{url}
\usepackage{verbatim}
\usepackage{graphicx}
\usepackage{cite}
\usepackage{comment}
\usepackage{amsthm}
\usepackage{bm}
\usepackage[colorlinks=true,linkcolor=blue]{hyperref}
\usepackage{xcolor,etoolbox}
\usepackage{tikz}
\theoremstyle{remark}
\usepackage{pgfplots}
\usetikzlibrary{calc,arrows.meta,positioning}
\usetikzlibrary{shapes.arrows}
\usetikzlibrary{fit,backgrounds} 
\usepackage{standalone}
\usepackage{mathrsfs}
\usepackage{amssymb}
\usepackage[justification=centering]{caption}
\usepackage{mathtools}
\usepackage{thmtools,thm-restate}
\usepackage[square, comma, sort&compress, numbers]{natbib}
\usepackage{enumerate}
\usepackage{setspace}
\usepackage{float}
\newcommand{\oplabel}[1]{\refstepcounter{equation}(\theequation\ltx@label{#1})}
\newtheorem{theorem}{Theorem}
\newtheorem{lemma}{Lemma}

\newtheorem{corollary}{Corollary}
\newtheorem{definition}{Definition}
\newtheorem{proposition}{Proposition}
\newcommand{\RNum}[1]{\uppercase\expandafter{\romannumeral #1\relax}}
	
\newcommand{\tb}[1]{{\textbf{#1}}}	
\normalsize
\begin{document}
	
	%
	\title{Excess Distortion Exponent Analysis for Semantic-Aware MIMO Communication Systems}
	%
	%
	%
	
	\author{
		Yuxuan Shi\thanks{Yuxuan Shi and Shuo Shao are with the School of Cyber and Engineering, Shanghai Jiao Tong University, Shanghai 200240, China (e-mail: ge49fuy@sjtu.edu.cn; shuoshao@sjtu.edu.cn).}, Shuo Shao, \IEEEmembership{Member,~IEEE}, Yongpeng Wu, \IEEEmembership{Senior Member,~IEEE}, Wenjun Zhang, \IEEEmembership{Fellow,~IEEE},\thanks{Yongpeng Wu and Wenjun Zhang are with the Department of Electronic Engineering, Shanghai Jiao Tong University, Shanghai 200240, China (e-mail: {yongpeng.wu, zhangwenjun}@sjtu.edu.cn).} \\Xiang-Gen Xia, \IEEEmembership{Fellow,~IEEE},\thanks{Xiang-Gen Xia is with the Department of Electrical, and Computer Engineering, University of Delaware, Newark, DE 19716, USA (e-mail: xianggen@udel.edu).} Chengshan Xiao, \IEEEmembership{Fellow,~IEEE} \thanks{Chengshan Xiao is with the Department of Electrical, and Computer Engineering, Lehigh University, Bethlehem, PA 18015, USA (e-mail: xiaoc@lehigh.edu).}}
	
	%
	%

\markboth{}%
{}
%



\maketitle
\begin{spacing}{1.5}
	\vspace{-1cm}
	\begin{abstract}
		In this paper, the analysis of excess distortion exponent for joint source-channel coding (JSCC) in semantic-aware communication systems is presented. By introducing an unobservable semantic source, we extend the classical results by Csiszar to semantic-aware communication systems. Both upper and lower bounds of the exponent for the discrete memoryless source-channel pair are established. Moreover, an extended achievable bound of the excess distortion exponent for MIMO systems is derived. Further analysis explores how the block fading and numbers of antennas influence the exponent of semantic-aware MIMO systems. Our results offer some theoretical bounds of error decay performance and can be used to guide future semantic communications with joint source-channel coding scheme.
	\end{abstract}
	

	%
	\IEEEpeerreviewmaketitle

	\section{Introduction}

	\IEEEPARstart{A}{s} a new paradigm in 6G networks, semantic communication gains significant attention in recent days, and is expected to become a promising technology in future wireless communications. According to the definition of semantic information from Weaver and Shannon \cite{Burks_Shannon_Weaver_1951}, this new paradigm, considering the meanings behind symbols instead of pursuing the accurate reconstructions, is able to transmit the desired semantic information to specific receivers. Consequently, compared with the conventional paradigm, semantic-aware communication systems can thus compress the source information in a larger extent, and reduce the corresponding communication cost, such as transmitting power and spectrum resources in wireless systems. Furthermore, for the future potential scenarios (e.g., smart Cities, IoT, virtual reality, etc.) whose main purposes are to enable the receiver know the intrinsic meanings and complete the specific tasks, the studies on semantic communication will be inevitably in full flourish. 
	
	\subsection{Related Works}
	The concept of semantic communication was given by the landmark work \cite{Burks_Shannon_Weaver_1951} in 1950s, in which the author conceived the communication over semantic level. Hereafter the efforts on how to model the semantic information in practical communication have been made in the last seven decades \cite{Carnap1953An,2011Towards,2011Universal}. Specifically, Carnap \cite{Carnap1953An} proposed the  logical probability measure for contexts instead of the statistical probability measure in Shannon's classic theory, Bao \cite{2011Towards} stressed that the background information plays a key role in the semantic communication, and Juba \cite{2011Universal} utilized the feedback/sensing as the intermediate to capture the essence of a message in a goal-oriented communication system. More recently, Liu, Zhang and Poor \cite{Liu_Zhang_Poor_2021} proposed a rate-distortion framework to characterize the semantic information, which models the source as intrinsic and extrinsic states, and solve the optimization problem in some special cases. Liu et. al. in \cite{Fangfang2022} extended the rate-distortion function according to the information bottleneck theory, and realize the semantic-aware image compression. The authors in \cite{JinhoChoi2022} connected a semantic communication layer (SC) on top of the technique communication layer (TC), and proposed different measures on entropy to enhance the knowledge base.
	
	Besides the aforementioned theoretical works, lots of papers focus on the practical realization of semantic communication with the help of artificial intelligence (AI). Numbers of frameworks on semantic communication were proposed to improve the compression or transmission performances, based on the machine learning techniques in terms of the texts, audios and images (see e.g., \cite{Nariman2018,Xie_Qin_Li_Juang_2020,Zhenzi2021,Kountouris_Pappas_2021,Huang2021,Dommel_Utkovski_Simeone_Stanczak_2021, Jiang_Wen_Jin_Li_2021} for a few representative works). Among these, joint source-channel coding (JSCC) based on deep learning (DL) networks is widely applied to improve the semantic communication performance. More specifically, authors in \cite{Xie_Qin_Li_Juang_2020} proposed a general DL-based JSCC framework for semantic communication systems, which is named as DeepSC. Based on the result in \cite{Xie_Qin_Li_Juang_2020}, authors in \cite{Zhenzi2021} presented a similar JSCC framework for a speech transmission and recognition. The authors in \cite{Dommel_Utkovski_Simeone_Stanczak_2021} and \cite{Jiang_Wen_Jin_Li_2021} extended the DeepSC framework in more practical scenarios, which combined the DL-based semantic communication with IoT fog networks and hybrid auto repeat quires (HARQ), respectively. 
	\subsection{Motivations and Contributions}
	Undoubtedly, semantic-aware communication provides a new paradigm of intelligent information exchanges in nowadays wireless communication networks. Nevertheless, existing theoretical works pay more attention to the compression but usually involve (or even not) simple channel models, which cannot offer meaningful guides for the implementations of JSCC-based semantic communication in practical 6G networks. Sparked by the above issue, it is natural to investigate the performance of JSCC-based semantic communication under practical wireless channels, e.g. multiple-input multiple-output (MIMO) channels with fadings, which shows fundamental limits for practical semantic communications. As a revolutionary technique in nowadays wireless networks, MIMO techniques benefit from the space multiplexing and obtain higher channel capacity. Numbers of researches focus on MIMO communication theory, such as capacities analysis \cite{winters1987capacity,telatar1999capacity}, channel diversity analysis \cite{dighe2003analysis,alamouti1998simple,Shin_Win_2008} and block coding regimes \cite{foschini1996layered,tarokh1998space}. Moreover, to verify the superiority of a JSCC scheme, error exponent is chosen as the performance measure, since separated source- channel coding (SSCC) performs the same as JSCC, in error probability sense with infinite block length, while JSCC is strictly optimal in error exponent sense. Roughly speaking, error exponent is the number $E$ with property that the error probability of a suitable code is $e^{-En}$ with block length $n$. Therefore, the error exponent can be used to measure the JSCC-based semantic communication performance. The explorations on error exponents of channel and source with fidelity criterion were given by Gallager \cite{Gallager_1968} and Marton \cite{Marton_1974}, respectively. Furthermore, Csiszar \cite{Csiszar_1981, Csiszar_1982} derived the error exponent of JSCC scheme, and presented it in a divergence form. Zhong \cite{Zhong_Alajaji_Campbell_2007,Zhong_Alajaji_Campbell_2009,Yangfan_Zhong_Alajaji_Campbell_2006} and Chang \cite{Chang_2009} extended the conclusion to systems with continuous alphabet and side information, respectively. The analysis on Gallager's random coding bound of MIMO channel exponent was stated in \cite{Shin_Win_2008,Alfano_Chiasserini_Nordio_Zhou_2015}.
	
	Inspired by the framework in \cite{Liu_Zhang_Poor_2021}, this paper considers a point-to-point semantic-aware communication system under JSCC framework. Specifically, following the rate-distortion function on characterizing the semantic information, we first start from a long Markov chain which consists of a source pair $(S,X)$, a noisy channel $W$ and the reconstructions $(\hat{S},\hat{X})$, in which $S$ represents the semantic source (intrinsic state) and $X$ stands for the observed source (extrinsic state). It is a generalized semantic communication model, which is named as semantic-aware communications, owing to the two necessary distortion constraints on semantic and observed reconstructions. This is the main difference between the remote source coding problem and our source system model. Next we emphasize that this model is highly consistent with most of the AI-based semantic communication works. Among these, some works transmit the extracted semantics and hope to recover the original texts/images/videos
	at the receiver \cite{Yang2022,Kountouris_Pappas_2021,Tian2022,Xie_Qin_Li_Juang_2020}, which means they consider the observed recovery $\hat{X}$ in their loss functions. Some other works execute the feature-specified tasks \cite{Nariman2018,Zhenzi2021,Huang2021, Jiang_Wen_Jin_Li_2021}, e.g., the object detection and image recognition, which means the semantic recovery $\hat{S}$ is considered. Then in the second part of this paper, we further generalize the model to a MIMO case, and obtain an achievable JSCC error exponent for a semantic-aware MIMO system. This extension enables the application of error exponent-optimal JSCC scheme in 6G wireless networks. Finally we conclude the main technical problems in the theoretical analysis: it is hard to characterize the joint typical sets of source sequences when we incorporate an extra semantic source. This obstacle is solved by introducing a channel coding theorem from \cite{Csiszar_1981} to show the joint typicality among the semantic, observed and received sequences. Moreover, to obtain the optimal exponent in MIMO systems, the random matrices instead of random scalars are operated, e.g., the integration of random channel state matrix, which is difficult to calculate. Hence, the hypergeometry function is utilized in the statement for further computation.
	
	Under this model, we first investigate the exponential rate of the excess distortion probability that either the recovered semantic or observed sequences exceed their required distortions (thus we use the notation ``excess distortion exponent'' instead of ``error exponent'' in the following). Upper and lower bounds of the exponent are presented as optimization problems in a discrete and memoryless case. We verify that our results can be degenerated to the Csiszar's JSCC exponent \cite{Csiszar_1982} or Weissman and Merhav's noisy source coding exponent \cite{Weissman_Merhav_2002,Weissman_2004}, and a direct conclusion is obtained that semantic-aware communication enlarges the error exponent in comparison with the conventional paradigm. Further, under a Gaussian source combined with a MIMO block fading channel, an achievable excess distortion exponent of JSCC schemes is given. In this case, the influences from coherence time, correlation coefficient and antennas numbers can be explicitly discussed. From the achievability bound, a list coding scheme can be designed by combining the list size with the semantic entropy. Moreover, the bound can extend some existing works on JSCC scheme for wireless communications, e.g., spatial coupled LDPC or D-polar codes \cite{Qiufang2022,Kai2022} to the semantic-aware scenarios. Besides, solution of the optimization problem of JSCC exponent for semantic-aware MIMO systems is offered. Finally, numerical results on the exponent are also presented to show how the environment parameters affect the exponent.
	
	This paper is organized as follows: in Section \ref{Sec2}, we give the notations on semantic-aware communication system, joint source-channel coding scheme and the excess distortion exponent. In Section \ref{Sec3}, upper and lower bounds on JSCC excess distortion exponent are presented in the discrete and memoryless case, as well as the degenerated cases to Csiszar, Weissman and Merhav's exponents. In Section \ref{Sec4}, a theory on achievable parametric form in a MIMO communication system and its optimization problem is presented. In Section \ref{Sec5}, we provide some examples and plots to illustrate the exponential behaviors of JSCC exponent, and discuss the influences of the key quantities.
	
	\section{Problem Formulation}	\label{Sec2}
	In this section, we present the model of the semantic-aware communication system, including the definitions of semantic-aware JSCC scheme, the excess distortion event and the excess distortion exponent.
	
	Throughout the paper, an upper case letter stands for a random variable, whose realization is represented by a lower case letter, and its alphabet is a calligraphy letter. For example, $x$ taking values in $\mathcal{X}$ is the realization of random variable $X$. $|\mathcal{X}|$ is the cardinality of $\mathcal{X}$, and $(x)^+$ denotes $\max(x,0)$. The distribution $P_X$ is the probability mass function (pmf) of $X$ if it has a countable alphabet. Besides, sequences are labeled with its length as superscript, such as $X^n = (X_{1},X_{2},\cdots,X_{n})$ and its realization $x^n$ follows similarly. $\mathbb{E}_{P(x)}[X]$ represents the expectation of random variable $X$ according to distribution $P(x)$, and $I_{P(x,y)}(X,Y)$ denotes the mutual information between $X$ and $Y$ in terms of joint distribution $P(x,y)$. $\mathcal{C}(\mathcal{A}\rightarrow\mathcal{B})$ denotes the set of all conditional distributions $P(b|a)$ where $a\in
	\mathcal{A}$ and $b\in\mathcal{B}$. Moreover, vectors and matrices are represented by bold letter, and $\tb{I}_m$ is the $m\times m$ identity matrix. Superscript $H$ and operator $\mathrm{tr}(\cdot)$ denote the transpose conjugate and trace function, respectively. Finally, $\tb{X}\in\mathbb{C}^{m\times n}\sim\mathcal{MN}(\tb{M},\tb{U},\tb{V})$ means that $\tb{X}$ follows matrix normal distribution with probability density function 
	\[
	p_{\tb{X}}(\tb{X})=\pi^{-m n} \operatorname{det}(\tb{U})^{-n} \operatorname{det}(\tb{V})^{-m} \exp\left\{\mathrm{tr}\left(-\tb{U}^{-1}(\tb{X}-\tb{M}) \tb{V}^{-1}(\tb{X}-\tb{M})^{H}\right)\right\},
	\]
	where $\tb{M} \in \mathbb{C}^{m \times n}, 0<\tb{U}=\tb{U}^{H} \in \mathbb{C}^{m \times m}$, $0<\tb{V}=\tb{V}^{H} \in \mathbb{C}^{n \times n}$, and $\tb{A}>0$ means that matrix $\tb{A}$ is positive definite.
	\subsection{Problem Formulation}
	\begin{figure}[h]
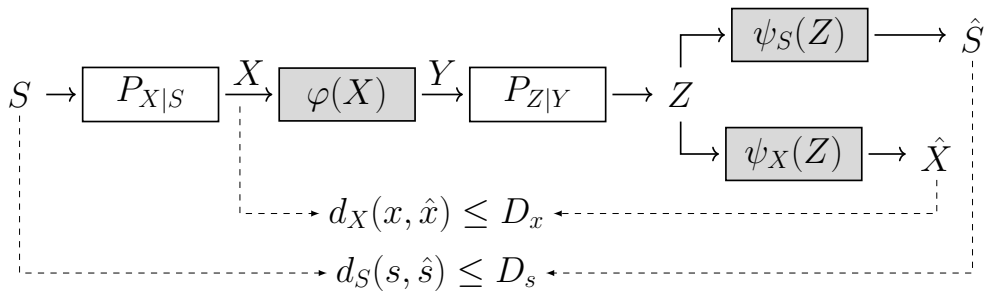

		\centering
		\includestandalone[width=0.8\textwidth]{figure/ChannelModel}%
		\caption{A semantic-aware communication system}
		\label{Sys}
	\end{figure}
	A semantic-aware communication system is depicted in Fig. \ref{Sys}. A discrete memoryless source (DMS) is described as a pair of random variables $(S,X)$ with joint distribution $P_{S,X}$ in product alphabet $\mathcal{S}\times\mathcal{X}$. In this model, $S$ is considered as the invisible intrinsic state with semantic information while $X$ is the extrinsic state and appears as the observable information. Moreover, a memoryless channel $W$ is defined with input $Y\in\mathcal{Y}$, output $Z\in\mathcal{Z}$ and transition probability $P_{Z|Y}$ (In the following, we denote the channel $W_{Z|Y}$ for simplicity). To introduce the block coding scheme, the probability mass function of a $k$-length independent and identically distributed (i.i.d.)
		sequence $s^k=(s_1,\cdots,s_k)\in\mathcal{S}^k$ is hence given by $P_{S^k}(s^k)=\prod_{i=1}^kP_S(s_i)$, and $P_{X^k|S^k}(x^k|s^k)=\prod_{i=1}^kP_{X|S}(x_i|s_i)$. For such a communication system, a joint source-channel code with block length $n$ and transmission rate $t=\frac{k}{n}$ symbol per channel use for the memoryless source $(S,X)$ and channel $W_{Z|Y}$ is defined as a tuple of mappings: \[
		\varphi^n(\cdot): \mathcal{X}^k\rightarrow\mathcal{Y}^n,
		\psi_S^k(\cdot): \mathcal{Z}^n\rightarrow\mathcal{\hat{S}}^k,
		\psi_X^k(\cdot): \mathcal{Z}^n\rightarrow\mathcal{\hat{X}}^k.
		\]
		That is, a $k$-length information block $s^k$ extracted from semantic source is observed as a $k$-length observed block $x^k$, and then is encoded through JSCC as a codeword $y^n=(y_1,y_2,\cdots,y_n)$ $=\varphi^n(x^k)$, transmitted, received as $z^n=(z_1,z_2,\cdots,z_n)$. Two different decoders decode the same received block as $\hat{s}^k=\psi^k_S(z^n)$ and $\hat{x}^k=\psi^k_X(z^n)$, corresponding to the desired semantic and observed information sequences, respectively. To measure the source distortion, we denote $d_S^k$ and $d_X^k$ the block-wise distortion measure functions of semantic and observable sources,
	\begin{align}
		&d_S: \mathcal{S}\times\mathcal{\hat{S}}\rightarrow\mathbb{R},\quad \hspace{0.2cm}d_S^k(s^k,\hat{s}^k) \triangleq \frac{1}{k}\sum_{i=1}^{k}d_S(s_i,\hat{s}_i),\label{sj}\\
		&d_X: \mathcal{X}\times\mathcal{\hat{X}}\rightarrow\mathbb{R},\quad d_X^k(x^k,\hat{x}^k) \triangleq \frac{1}{k}\sum_{i=1}^{k}d_X(x_i,\hat{x}_i),\label{jj}
	\end{align}
	where $\hat{s}^k=(\hat{s}_1,\hat{s}_2,\cdots,\hat{s}_k)\in\mathcal{\hat{S}}^k$ and $\hat{x}^k=(\hat{x}_1,\hat{x}_2,\cdots,\hat{x}_k)\in\mathcal{\hat{X}}^k$ represent the recovered semantic and observed sequences, respectively. Given a source pair $(S,X)$ and a channel $W_{Z|Y}$, and two distortions $D_s,D_x\geq0$ on semantic and observed sequences, respectively, we define the erroneous set of $(s^k,x^k,z^n)$ that violates the distortion constraints as
	\begin{align}
		\mathcal{E} &= \Big\{{\left(s^k,x^k,z^n\right)\in\mathcal{S}^k\times\mathcal{X}^k\times\mathcal{Z}^n}: {d^k_{S}\left(s^k,\psi^k_S\left(z^n\right)\right)> D_s\text{ or }d^k_{X}\left(x^k,\psi^k_X\left(z^n\right)\right)> D_x}\Big\}\notag.
	\end{align}
Note that in remote source coding, only the indirect source is concerned, while both indirect and direct sources are recovered in a semantic-aware system. Hence how semantic distortions affect the coding scheme performance, and the tradeoff between semantic and observed distortions can be discussed. Therefore, we define a lossy JSCC scheme for semantic-aware communications which is able to recover both semantic and observable information as the following.
	\begin{definition}[Lossy Joint Source-Channel Code for semantic-aware communications]\label{JSCCdef}
		The tuple $(\varphi^n,\psi_S^k,\psi_X^k)$ is an $(n,k,D_s,D_x)$ lossy joint source-channel code for semantic source $S\in{\mathcal{S}}$, observable source $X\in\mathcal{X}$ and memoryless channel $W_{Z|Y}$ with two distortions $D_s,D_x\geq0$ if $\mathbb{P}\{\mathcal{E}\}\leq \epsilon$, where $\epsilon$ is a sufficient small positive number. The code rate $R=\frac{1}{n}\log |\mathcal{Y}^n|$.
	\end{definition}
	The JSCC excess distortion probability can be stated as
	\begin{align}
		\mathbb{P}\left\{\mathcal{E}\right\}	\triangleq\sum_{x^k\in\mathcal{X}^k}P_{X^k}\left(x^k\right)\sum_{s^k\in\mathcal{S}^k}P_{S^k|X^k}\left(s^k|x^k\right)\sum_{z^n\in\mathcal{E}(s^k,x^k)}P_{Z^n|Y^n}\left(z^n|\varphi^n\left(x^k\right)\right),\label{statement}
	\end{align}
	where 
	$\mathcal{E}\left(s^k,x^k\right)=\left\{z^n\in\mathcal{Z}^n:\left(s^k,x^k,z^n\right)\in\mathcal{E}\right\}.$
	
	Here we use summation if the alphabets are finite, for continuous source and channel pairs, Eq. \eqref{statement} can be rewritten by replacing the summation with the integration. The following definition introduces the JSCC excess distortion exponent.
	\begin{definition}\label{EJ}
		The optimal  JSCC excess distortion exponent  $E^{\mathrm{opt}}_J(P_{X},P_{S|X},W_{Z|Y},D_s,D_x,t)$ for any $D_s,D_x\geq 0$, is defined as the supremum of the set including all numbers $E$ for which there exists a sequence of $(n,k,D_s,D_x)$ JSCC scheme such that
		\begin{align}
			E\leq\liminf_{n\rightarrow\infty}\left[-\frac{1}{n}\log \mathbb{P}\left\{\mathcal{E}\right\}\right].
		\end{align}
	\end{definition}
	In the following, we try establishing upper and lower bounds on this excess distortion exponent $E^{\mathrm{opt}}_J(P_{X},P_{S|X},W_{Z|Y},D_s,D_x,t)$ in the case of discrete and memoryless source-channel pair.
	\section{Joint Source-Channel Coding Excess Distortion Exponent for Semantic-Aware Communications}\label{Sec3}
	In this section, we first investigate bounds on JSCC excess distortion exponent of a discrete and memoryless semantic-aware communication system depicted in Fig. \ref{Sys}. The bounds are composed of the source exponent and the channel exponents. We then verify that the proposed bounds can be degenerated to the known results if relax one of the distortion constraints. 
	\subsection{Statement of the Main Result}
	\begin{restatable}{theorem}{Maintheorem}\label{Maintheorem}
		For a given memoryless observable source $X$ with distribution $P_X$, a conditional distribution $P_{S|X}$, a memoryless channel with transition probability $W_{Z|Y}$, and two distortions $D_s,D_x\geq0$, which satisfies $tR(P_X,P_{S|X},D_x,D_s)$ $\leq$ $C(W_{Z|Y})$, the excess distortion exponent $E_J^{\mathrm{opt}}\left(P_X,P_{S|X},W_{Z|Y},\right.$$\left.D_s,D_x,t\right)$ for optimal $(n,k,D_s,D_x)$ JSCC with distortions $D_s,D_x$ and transmission rate $t$ is bounded by
		\begin{align}
			E_J^{\mathrm{opt}}\left(P_X,P_{S|X},W_{Z|Y},D_s,D_x,t\right)\leq 
			&\min_{R\in\mathcal{R}}\left\{t\widetilde{E}\left(\frac{R}{t},P_X,P_{S|X}\right)+E_{\mathrm{sp}}\left(R,W_{Z|Y}\right)\right\},\label{Upp}\\
			E_J^{\mathrm{opt}}\left(P_X,P_{S|X},W_{Z|Y},D_s,D_x,t\right)\geq &\min_{R\in\mathcal{R}}\left\{t\widetilde{E}\left(\frac{R}{t},P_X,P_{S|X}\right)+E_{\mathrm{ex}}\left(R,W_{Z|Y}\right)\right\}.\label{Low}
		\end{align}
		Herein
		\begin{align}
			\mathcal{R}&\triangleq\{R:tR(P_X,P_{S|X},D_x,D_s)\leq R\leq C(W_{Z|Y})\},\notag\\
			\widetilde{E}(r,P_X,P_{S|X})&\triangleq\min_{Q_X}\min_{\substack{U_{S|X}\in\mathcal{C}(\mathcal{X}\rightarrow\mathcal{S}):\\R(Q_X,U_{S|X},D_s,D_x)\geq r}}\left\{D(Q_X||P_X)+D(U_{S|X}||P_{S|X}|Q_X)\right\},\label{sourcecoding}\\
			E_{\mathrm{sp}}(R,W_{Z|Y})&\triangleq\max_{P_Y}\min_{V_{Z|Y}:I_{P_Y\times V}(Z;Y)\leq R}D(V_{Z|Y}||W_{Z|Y}|P_{Y}(Q_X))\label{rcex},\\
			E_{\mathrm{ex}}(R,W_{Z|Y})&\triangleq\max_{P_Y}\min_{P_{Y\widetilde{Y}}:P_{\widetilde{Y}}=P_{Y}}\left\{{\mathbb{E}d_{W_{Z|Y}}(Y,\widetilde{Y})}+I_{P_{Y\widetilde{Y}}}\left(Y;\widetilde{Y}\right)-R\right\}\label{channelex},
		\end{align}
		where $ D(V||W|P)\triangleq\sum_{y\in\mathcal{Y}}P(y)D\left(V(\cdot|y)||W(\cdot|y)\right)\label{Cond}$ denotes the conditional K-L divergence and $d_{W_{Z \mid Y}}(y, \widetilde{y})$ is named as the Bhattacharya distance between two channel inputs \cite[Chp~7]{Gallager_1968}. $C(W_{Z|Y})$ is the channel capacity and the rate distortion function characterizing semantic information is given by \cite[Thm~1]{Liu_Zhang_Poor_2021} as
		\begin{align}
			R(P_X,P_{S|X},D_s,D_x)&\triangleq\min_{P_{\hat{S},\hat{X}|X}}I(X;\hat{S}\hat{X}),\\
			\text{s.t.}\qquad\mathbb{E}[\hat{d}_S(X,\hat{S})]&\leq D_s\label{s},\\
			\mathbb{E}[d_X(X,\hat{X})]&\leq D_x,\label{x}
		\end{align}
		where $\hat{d}_S(x,\hat{s})=\mathbb{E}[d_S(S,\hat{s})|X=x]$, while $d_S(\cdot,\cdot)$ and $d_X(\cdot,\cdot)$ denote the component-wise distortion functions given in Eq. \eqref{sj} and Eq. \eqref{jj}, respectively.
	\end{restatable}
	\begin{proof}
		See Appendix \ref{theorem1}.
	\end{proof}
	In Theorem \ref{Maintheorem}, upper and lower bounds of the JSCC excess distortion exponent are established. The upper bound consists of the sphere-packing bound on channel error exponent $E_{\mathrm{sp}}(R,W_{Z \mid Y})$ and the source excess distortion exponent $\widetilde{E}(r,P_X,P_{S|X})$, which considers a new fidelity on semantic source. Meanwhile the lower bound is composed of the expurgated random coding bound on channel error exponent $E_{\mathrm{ex}}(R,W_{Z \mid Y})$ and the same source exponent. Note that our result is a generalized form of Csiszar's error exponent \cite{Csiszar_1982} in which a lossy JSCC encodes a single source $X$ and imposes a unique constraint on it. The basic idea to prove the results in Theorem \ref{Maintheorem} is as follows. To obtain the source exponent, we characterize the excess distortion probability in terms of sources, by counting the numbers of typical semantic and observable sequences. The joint typicality among the semantic, observed and received sequences is necessary to be discussed. To obtain the channel exponent, we prove the sphere-packing bound and the expurgated random coding bound still hold for the semantic-aware transmission in Fig. \ref{Sys} via Csiszar's channel coding theorem \cite{Csiszar_1981}. Finally, by minimizing the source and channel exponents jointly over a group of JSCC schemes, we formulate upper and lower bounds on optimal JSCC excess distortion exponent in Eq. \eqref{Upp} and Eq. \eqref{Low}, respectively.
	
	\begin{figure} 
		\captionsetup[subfigure]{margin=120pt} 
		\begin{minipage}[t]{0.46\linewidth}
			\subfigure[]{
				\centering
				\includegraphics[width=1\textwidth]{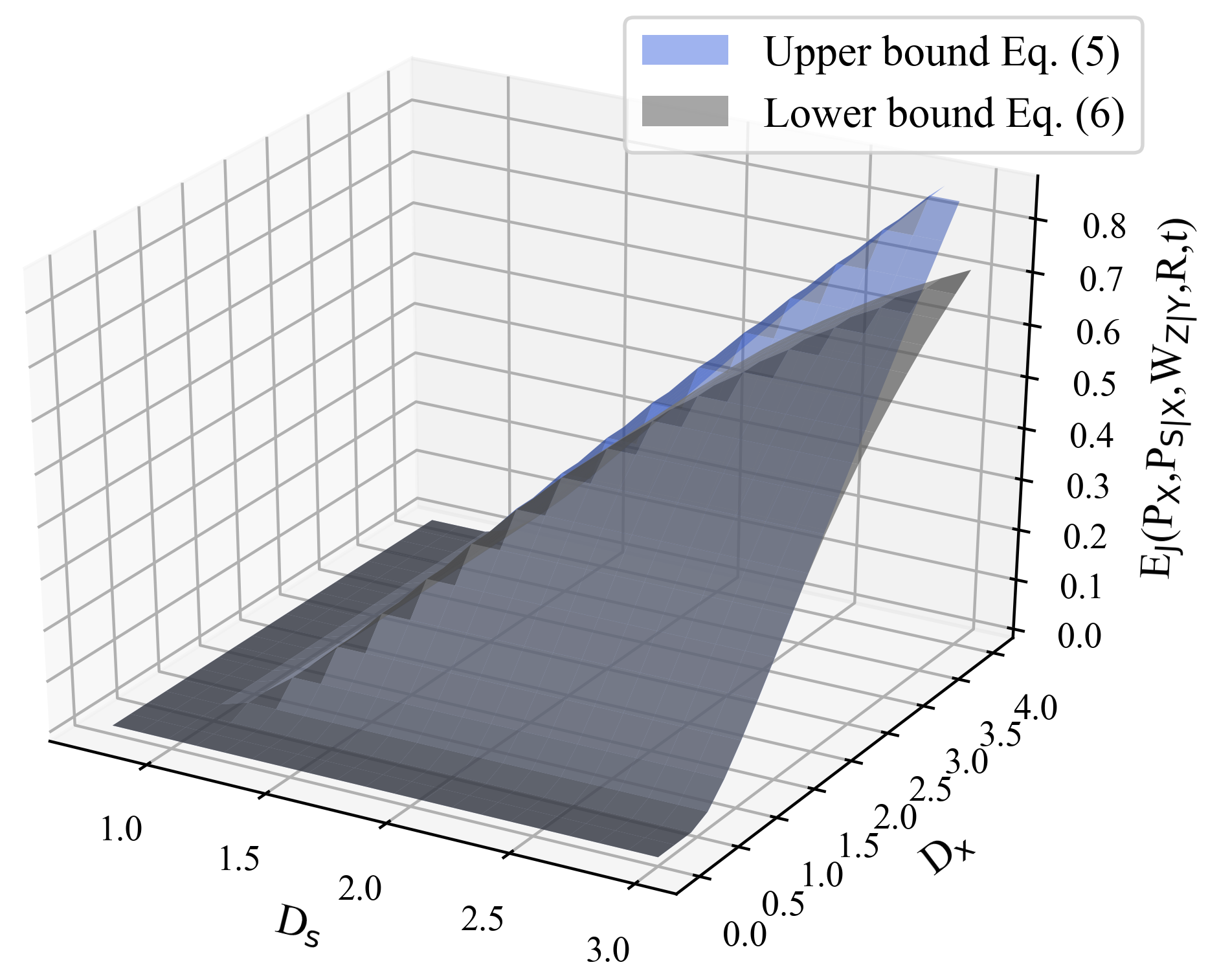}}
		\end{minipage}
		\hspace{0.1in}
		\begin{minipage}[t]{0.49\linewidth}
			\centering
			\subfigure[]{
				\label{3DSource}
				\includegraphics[width=0.95\textwidth]{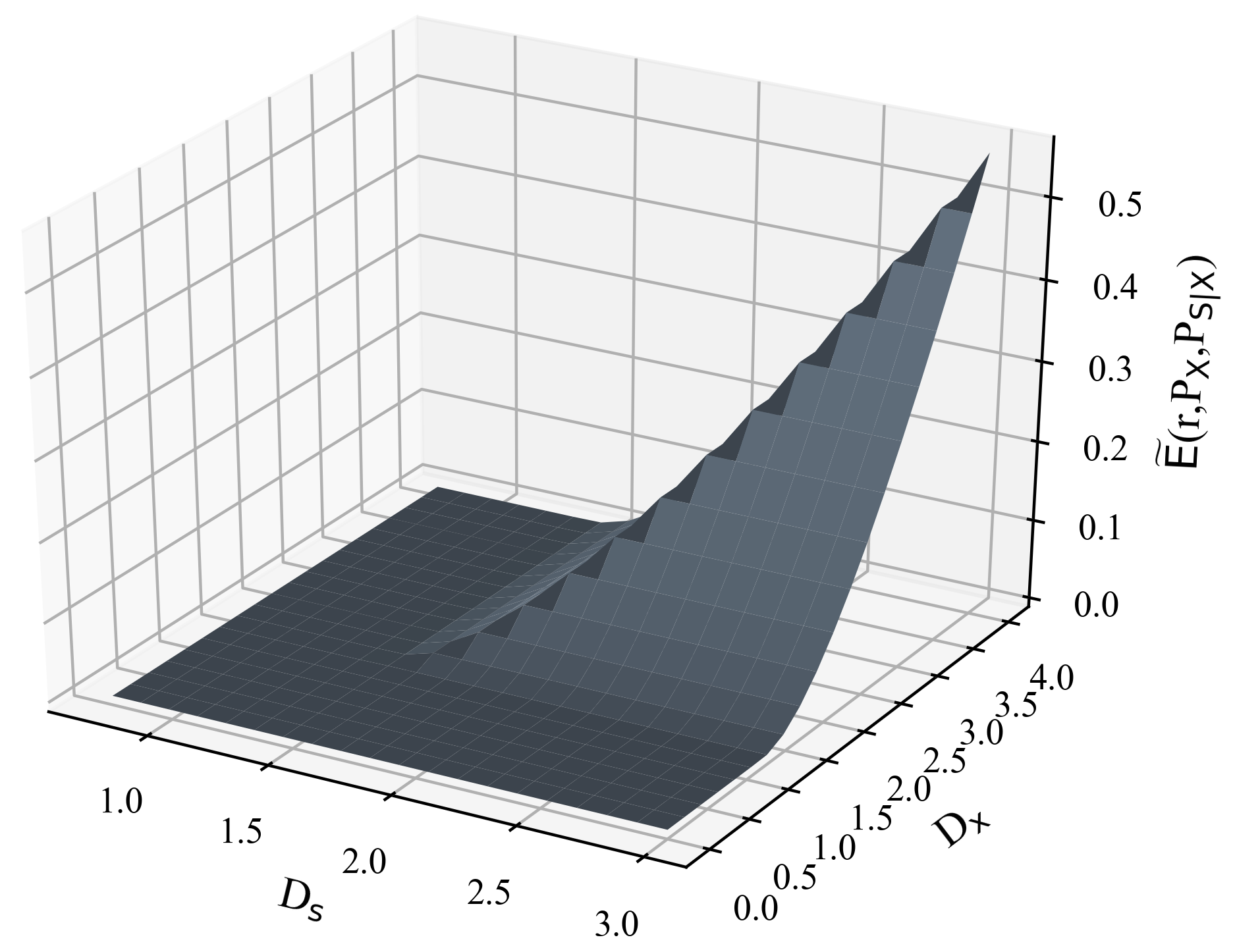}
			}
		\end{minipage}
		\caption{Characterization of the excess distortion exponent in a toy case (a) Upper and lower bounds in Eq. \eqref{Upp} and Eq. \eqref{Low}; (b) Source excess distortion exponent in Eq. \eqref{sourcecoding}}\label{3DJSC}
	\end{figure}
	We present an example of excess distortion exponent under a semantic-aware communication system in Fig. \ref{3DJSC}, in which a toy case is considered that semantic source takes value in $\mathcal{S}=\{1,2,3\}$ with equal probability, $\mathcal{X}=\{0,1\}$ and a binary symmetric channel (BSC) with flip rate $p=0.3$. The left hand side one plots the upper and lower bounds of JSCC exponent, while the right hand side one gives the source exponent in Eq. \eqref{sourcecoding}. In this case, it shows that the exponent turns to be a non-decreasing function over semantic and observed distortions. Nevertheless, the behavior of the exponent for generalized $(S,X)$ and $W_{Z|Y}$ is unpredictable.
	\subsection{Two Degenerated Cases on Semantic Source}
	In this subsection, we study the special case where $D_s=\infty$ or $D_x=\infty$, which means the constraint on semantic or observed information is relaxed. As follows, we verify that the bounds in Theorem \ref{Maintheorem} can be reduced to the known results under simpler settings, i.e., Csiszar's JSCC error exponent and Weissman's error exponent for noisy source coding \cite{Weissman_Merhav_2002}.
	\begin{corollary}[Csiszar's JSCC error exponent]\label{Csiszar}
		Without the semantic constraint, i.e., $D_s=\infty$, the communication system focuses on reconstructing the observed information $X$. Consequently, the excess distortion exponent of an optimal ($n,k,\infty,D_x$)-JSCC scheme, with the absence of achievable semantic constraint Eq. \eqref{s}, is reduced to the Csiszar's JSCC error exponent with a fidelity criterion \cite[Thm~2, Thm~4]{Csiszar_1982}.
	\end{corollary}
	\begin{proof}
		This corollary can be obtained intuitively due to the neglect of semantic information. For a rigorous proof, according to the excess distortion exponent of semantic source given in Eq. \eqref{sourcecoding}, the conditional divergence can be rewritten as
		\begin{align}
			D\left(U_{S|X}||P_{S|X}|Q_X\right)&=D\left(Q_X\times U_{S|X}||Q_X\times P_{S|X}\right)
			=
			\left\{
			\begin{aligned}
				& 0,  \quad\quad\text{if    }Q_X\times U_{S|X}=Q_X\times P_{S|X}\\
				&\infty,\quad\text{   otherwise}
			\end{aligned}
			\right.,\notag
		\end{align}
		which implies the infimum of this term can be obtained by choosing $U_{S|X}=P_{S|X}$ such that $\min_{U_{S|X}}D\left(U_{S|X}||P_{S|X}|Q_X\right)=0$.
	\end{proof}
	Comparing Theorem \ref{Maintheorem} and Corollary \ref{Csiszar}, it can be observed that the JSCC error exponent  for semantic-aware systems contains an extra non-negative term and is larger than the Csiszar's error exponent in non-trivial cases, which is translated as a faster error decay speed owing to the involvement of the semantic information.
	\begin{corollary}[Weissman and Merhav's Error Exponent]
		Without the constraint on observed information, i.e., $D_x=\infty$, the model is degenerated to an indirect source coding and communication system. Then Eq. \eqref{sourcecoding} can be reduced to Weissman and Merhav's error exponent \cite[Th~1]{Weissman_Merhav_2002}.
	\end{corollary}
	\begin{proof}
		Note that in this case, we focus on a noisy source excess distortion exponent, where the reconstruct semantic sequence $\hat{s}^k$ is deterministic for fixed $x^k$ and the compression scheme, since $\hat{s}^k=\psi_S^k\left(\varphi^k\left(x^k\right)\right)$. Given the sequences $\left(s^k,x^k\right)$ and the excess distortion event $\{d_S^k\left(s^k,\hat{s}^k\right)>D_s\}$, the rate-distortion function becomes
		\begin{align}
			R(P_X,P_{S|X},D_s,\infty)\triangleq\min_{P_{\hat{S}|X}}I(X;\hat{S}),\qquad
			\text{s.t.}\qquad\mathbb{E}[\hat{d}_S(X,\hat{S})]\leq D_s.\notag
		\end{align}
		The excess distortion is established directly between the transmitted $x^k$ and the reconstructed sequence $\hat{s}^k$. Thus, given $x^k$ of type $Q_X$, the test channel $U_{S|X}$ can be rewritten as:
		\begin{align}
			&U_{S|X} = \mathscr{W}\times\mathscr{V},\notag\\
			&\mathscr{W}\in\mathcal{C}\left(\mathcal{X}\rightarrow\mathcal{\hat{S}}\right):r\geq I_{Q_X\times\mathscr{W}}\left(X,\hat{S}\right)\notag,\\
			&\mathscr{V}\in\mathcal{C}\left(\hat{\mathcal{S}}\times\mathcal{X}\rightarrow\mathcal{S}\right):\mathbb{E}_{Q_X\times\mathscr{W}\times\mathscr{V}}d_S\left(S,\hat{S}\right)>D_s\notag,
		\end{align}
		in order to impose the constraints on mutual information and the distortion threshold, which directly yields the conclusion in \cite{Weissman_Merhav_2002}.
	\end{proof}
	\section{Achievable Excess Distortion Exponent in Semantic-Aware MIMO Systems}\label{Sec4}
	It is observed that Theorem \ref{Maintheorem} is easy to be extended into different scenarios. In wireless networks, it may be of interest to consider a MIMO block fading channel rather than a simple DMC. In this section, we present an achievable statement of JSCC excess distortion exponent for a semantic-aware MIMO system composed of Gaussian distributed semantic source and block fading MIMO channel. Furthermore, the optimization problem of exponent is solved in a simple case for explicit analysis.
	
	A proposition is used to claim the extension of Theorem \ref{Maintheorem} to a more general case.
	\begin{proposition}\label{p1}
		For a communication system, which consists of a joint Gaussian vector pair $(\tb{S},\tb{X})$ and a Gaussian channel $\tb{W}_{\tb{Z}|\tb{Y}}$ with arbitrary memory, the lower bound in Eq. \eqref{Low} holds. 
	\end{proposition}
	\begin{proof}
		Note that we considered discrete sources and channel with finite input and output alphabets before. Nevertheless, in \cite[Ch~4]{Zhong08} Zhong etc. combines Csiszar's exponents with Gallager's reliability functions in discrete cases via Fenchel duality, in which Gallager's statements are verified to be powerful tools on error exponent and are easily extended to the case of continuous-alphabet with arbitrary memory. It is worth mentioning that in Gaussian case, only random coding bound in Eq. \eqref{Low} can be extended and the sphere-packing bound remains unclear. For details, the reader can turn to \cite{Zhong_Alajaji_Campbell_2007,Zhong_Alajaji_Campbell_2009} for a rigorous proof of JSCC error exponent under the setting of Gaussian distributed system and a one-order Markovian system, respectively. Hence, by slightly abusing the notation defined in Sec. \ref{Sec2}, we can easily obtain this conclusion.
	\end{proof}
		Note that the lower bound in Eq. \eqref{Low} can be extended into a MIMO communication setting. However, even not consider the semantic source, the converse proof on error exponent is not easy to obtain in such a MIMO case. Among these proofs, Fano inequality and hypothesis testing show unavailable bounds on exponent, while sphere-packing bound in Eq. \eqref{Upp} though yields a tight bound in single antenna case, it is difficult to be extended into the multi-antennas case due to the following two aspects. First, for the sphere-packing of codeword, the solid angle of the Voronoi regions for matrices in continuous alphabet is difficult to characterize, hence the overall cone is not a circular cone. Second, the strong converse for JSCC in Lemma \ref{JSCCTH} does not necessarily hold in a MIMO communication system, since we cannot use the weak law of large numbers (WLLN) for memory case \cite[Appendix 1]{Zhong_Alajaji_Campbell_2009}. In summary, we present an achievability bound as follows, which reveals an achievable excess distortion exponent of JSCC for semantic-aware MIMO systems.
	\begin{theorem}\label{Gaussian}
		For the above semantic-aware communication system, the observed source $\tb{X}$ follows a Gaussian vector distribution $\mathcal{N}(\tb{0}_{q},\bm{\Sigma}_X)$, where $\bm{\Sigma}_X$ is an $q\times q$ positive semi-definite matrix. Meanwhile the $\ell$-length semantic source $\tb{S}$ is given by 
		\[
		\tb{S} = \tb{h}\tb{X}+\tb{N},
		\]
		where $\tb{h}$ is an $\ell\times q$ matrix, and $\tb{N}$ is a random vector follows $\mathcal{N}(\tb{0}_{\ell},\bm{\Sigma}_N)$, which is independent of $\tb{X}$. The quadratic distortion measures become $d_S(\tb{s},\hat{\tb{s}})=\mathrm{tr}\{(\tb{s}-\hat{\tb{s}})(\tb{s}-\hat{\tb{s}})^H\}$ and $d_S(\tb{x},\hat{\tb{x}})=\mathrm{tr}\{(\tb{x}-\hat{\tb{x}})(\tb{x}-\hat{\tb{x}})^H\}$, where $\hat{\tb{s}}$ and $\hat{\tb{x}}$ refer to the recovered source vectors. Furthermore, a MIMO communication system contains $n_T$ transmit and $n_R$ receive antennas, where the block fading channel $\tb{W}_{\tb{Z}|\tb{Y}}$ remains invariant for $N_c$ symbols in each coherence time. In each observation composed of $N_b$ independent coherence intervals, which amount to $N_bN_c$ symbols, the received matrix  $\tb{Z}_i\in\mathbb{C}^{n_R\times N_c}$ at the $i$-th interval can be formulated as
		\begin{align}
			\tb{Z}_i=\tb{H}_i\tb{Y}_i+\tb{W}_i,\quad i=1,2,\cdots,N_b\notag,
		\end{align} 
		where $\tb{Y}_i\in\mathbb{C}^{n_T\times N_c}$ is the channel input matrix, $\tb{H}_i$ is the channel state matrix and $\tb{W}_i$ is the additive white Gaussian noise matrix, namely $\tb{W}_i\sim\mathcal{MN}(\tb{0}_{n_R\times N_c},N_w\tb{I}_{n_R},\tb{I}_{N_c})$. Here $N_w$ is the noise coefficient. The transition probability with perfect CSI at the receiver can be stated as
		\begin{align}
			p(\tb{Z}|\tb{Y},\tb{H})=(\pi N_w)^{-n_RN_c}\exp\left\{-\frac{1}{N_w}(\tb{Z}-\tb{H}\tb{Y})(\tb{Z}-\tb{H}\tb{Y})^H\right\},\label{MIMOtran}
		\end{align}
		where the subscript $i$ is dropped for simplicity since the channel is memoryless for each coherence interval. Note that $\tb{Y}$ denotes the power constrained input as $\frac{1}{N_c}\mathbb{E}[\mathrm{tr}\{\tb{Y}\tb{Y}^H\}]=\mathrm{tr}\{\tb{Q}\}\leq \mathcal{P}$. Then, there exists an $(n,k,D_s,D_x)$-lossy JSCC for this semantic-aware MIMO communication system with excess distortion exponent
		\begin{align}
			E_J^{\mathrm{opt}}\left(\tb{P}_\tb{X},\tb{P}_{\tb{S}|\tb{X}},\tb{W}_{\tb{Z}|\tb{Y}},D_s,D_x,t\right)&=\min_{R\in\mathcal{R}}E_J\left(\tb{P}_\tb{X},\tb{P}_{\tb{S}|\tb{X}},\tb{W}_{\tb{Z}|\tb{Y}},R,t\right)\notag\\
			&=\min_{R\in\mathcal{R}}\left\{t\widetilde{E}^\mathrm{G}\left(\frac{R}{t},\tb{P}_\tb{X},\tb{P}_{\tb{S}|\tb{X}}\right)+E^{\mathrm{MIMO}}\left(R,\tb{W}_{\tb{Z}|\tb{Y}}\right)\right\},\label{MMo}
		\end{align}
		where 
		\begin{align}
			\widetilde{E}^\mathrm{G}(r,\tb{P}_\tb{X},\tb{P}_{\tb{S}|\tb{X}})=&\min_{\bm{\Delta}:\bm{O}\prec\bm{\Delta}\preceq\bm{\Sigma_X},\atop\mathrm{tr}\{\bm{\Delta}\}\leq D_x}\min_{\tb{A}\in\mathbb{C}^{q\times q}:\atop\det(\tb{A})=\det(\bm{\Delta})e^{2r}}\min_{\tb{B}\in\mathbb{C}^{\ell\times\ell}:\atop\mathrm{tr}\left\{\tb{B}\right\}=D_s-\mathrm{tr}\{\tb{h}^H\bm{\Delta}\tb{h}\}}\Bigg\{\frac{1}{2}\log\frac{\det(\bm{\Sigma}_X)\det(\bm{\Sigma}_N)}{\det(\bm{\Delta})e^{2r}\det(\tb{B})}\notag\\
			&+\mathrm{tr}\left\{\bm{\Sigma}_X^{-1}\tb{A}+\bm{\Sigma}_N^{-1}\tb{B}+\left(\bm{\Sigma}_N^{-1}\tb{B}-\tb{I}_{\ell}\right)\tb{h}^H\bm{\Sigma}_X\tb{h}\right\}\Bigg\}-\ell-q,\label{Gsource}
		\end{align}
		\begin{align}
			E^{\mathrm{MIMO}}\left(R,\tb{W}_{\tb{Z}|\tb{Y}}\right)=&\max _{0 \leq \rho \leq 1}\left\{\max _{\delta \geq 0}E_\mathrm{ex}(\tb{Q},\rho,\delta,N_c) -\rho R\right\},\\
			E_\mathrm{ex}\left(\tb{Q}, \rho, \delta, N_{\mathrm{c}}\right)=&2\delta\rho\mathcal{P}\notag\\
			&-\frac{1}{N_c}\ln\mathbb{E}_{\textbf{H}}\left\{\det{\left(\tb{Q}\tb{A}\left(\textbf{I}_{n_T}-\tb{Q}\left(\frac{\tb{H}^H\tb{H}\tb{A}^{-1}\tb{H}^H\tb{H}}{16N_w^2\rho^2}-\frac{\tb{H}^H\tb{H}}{4N_w\rho}+\delta\right)\right)\right)^{-N_c\rho}}\right\}\label{exb},\\
			\tb{A}=&\delta\tb{I}_{n_T}-\tb{Q}^{-1}-\frac{1}{4N_w\rho}\tb{H}^H\tb{H}\label{nb}.
		\end{align}
	\end{theorem}
	\begin{proof}
		See Appendix \ref{Gauss}.
	\end{proof}
	In this theorem, we present an achievable JSCC excess distortion exponent in a semantic-aware MIMO communication system. Specifically, we consider a jointly Gaussian distributed source pair, where an $\ell$-length semantic vector $\tb{S}$ is combined with a $q$-length observed vector $\tb{X}$. Furthermore, the quadratic distortion measure and a MIMO system with block fading channel are also considered. The optimal exponent $E_J^{\mathrm{opt}}$ is composed of two parts, namely the source exponent  and the expurgated random coding exponent for MIMO systems. From the achievable bound, the ergodic capacity and cut-off rate of the above semantic-aware MIMO communication system can be obtained, by setting $\rho=0$ or 1. We note that the generalized vector nature complicates the statement of the exponent, which makes the  optimization in Eq. \eqref{MMo} difficult. Hence a simple case is discussed as follows, which enables us to further analyze and reveal some insights on the JSCC scheme design in such a semantic-aware MIMO communication system, for optimal excess distortion exponent.
	
	\begin{corollary}[Excess distortion exponent for semantic-aware MIMO systems in specific case]\label{Cor4}
		Under the same setups in Theorem \ref{Gaussian}, we further assume $\tb{X}\sim\mathcal{N}(\tb{0}_q,\sigma_X^2\tb{I}_q)$,  $\tb{N}\sim\mathcal{N}(\tb{0}_\ell,\sigma_N^2\tb{I}_\ell)$ and $\tb{H}\sim\mathcal{MN}(\tb{0}_{n_R\times n_T},\tb{I}_{n_R},\tb{I}_{n_T})$ (for simplicity we set $n_R\leq n_T$), $\tb{Q}=\frac{\mathcal{P}}{n_T}\tb{I}_{n_T}$ due to the equal power assignment on the transmitting antennas, and denote $\mathrm{SNR}=\frac{\mathcal{P}}{N_w}$, then the following equations hold:
		\begin{enumerate}[(a)]
			\item For the expurgated random coding bound for MIMO channel in Eq. \eqref{nb}, the derivatives are calculated as
			\begin{align}
				\frac{\partial E_\mathrm{ex}\left(\tb{Q}, \rho, \delta, N_{\mathrm{c}}\right)}{\partial \delta}=&2\rho\mathcal{P}-\frac{2\rho n_T\mathrm{SNR} }{1-\delta\text{ } \mathrm{SNR}}+\frac{\mathrm{SNR}}{N_c}\mathrm{tr}\left\{\tb{K}^{-1}(\rho,\delta)\frac{\partial \tb{K}(\rho,\delta)}{\partial\delta}\right\}\label{delta}\\
				\frac{\partial E_\mathrm{ex}\left(\tb{Q}, \rho, \delta, N_{\mathrm{c}}\right)}{\partial \rho} =& 2\delta\mathcal{P}+2n_T\ln(1-\delta\text{ }\mathrm{SNR})-\frac{1}{N_c}\mathrm{tr}\left\{\tb{K}^{-1}(\rho,\delta)\frac{\partial \tb{K}(\rho,\delta)}{\partial\rho}\right\}\label{rho}
			\end{align}
			where $\tb{K}(\rho,\delta)$ is a Hankel matrix with size $n_T\times n_T$ whose $(i,j)$-th entry follows hypergeometric function 
			\[
			(n_T-n_R+i+j-2)!_2F_0\left(n_T-n_R+i+j-1,N_c\rho;-\frac{\mathrm{SNR}}{2(1-\delta\text{ } \mathrm{SNR})\rho}\right)
			\]
			\item Given $tR(\tb{P}_\tb{X},\tb{P}_{\tb{S}|\tb{X}},D_x,D_s)\leq R\leq C(\tb{W}_{\tb{Z}|\tb{Y}})$, the JSCC excess distortion exponent is convex in terms of code rate $R$.
			\item Given the above source-channel pair, and the transmission rate $t$, the optimal achievable code rate $R^\star$ can be formulated by
			\begin{align}
				R^\star = 2t\ln\frac{t\rho^\star+2}{2\min\left\{\frac{\sigma_X^2}{\sigma_X^2\mathrm{tr}\{\tb{h}^T\tb{h}\}+\sigma_N^2}(D_s-\sigma_N^2),\frac{1}{\sigma_X^2}D_x\right\}}
			\end{align}
			where $\rho^\star$ satisfies \[
			\rho^\star=\arg\max _{0 \leq \rho \leq 1}\left\{E_\mathrm{ex}(\tb{Q},\rho,\delta,N_c) -\rho R\right\}
			\] according to Eq.\eqref{delta} and Eq.\eqref{rho}.
		\end{enumerate}
	\end{corollary}
	\begin{proof}
		Given $\tb{Q}=\frac{\mathcal{P}}{n_T}\tb{I}_{n_T}$ and Gaussian distributed random matrix $\tb{H}$, (a) is obtained by
		\begin{align}
			E_\mathrm{ex}\left(\frac{\mathcal{P}}{n_T}\tb{I}_{n_T}, \rho, \delta, N_{\mathrm{c}}\right)=2\delta\rho\mathcal{P}-2\rho n_T\ln(1-\delta\text{ }\mathrm{SNR})-\frac{1}{N_c}\ln\frac{\det(\tb{K}(\rho,\delta))}{\mathcal{K}}
		\end{align}
		where $\tb{K}$ is given above and $\mathcal{K}=\prod_{i=1}^{n_T}(n_R-i)!(i-1)!$ according to \cite[Cor~4]{Shin_Win_2008}. This corollary enables us to obtain the expectation in Eq. \eqref{nb} via the computable generalized hypergeometric function $_2F_0(\cdot,\cdot;\cdot)$ (given by \cite[Eq~(3)]{James_1964}). Hence the partial derivatives are stated in Eq. \eqref{delta} and Eq. \eqref{rho}. For (b), the conclusion is also direct since the second partial derivative over code rate $R$ is positive when the code rate lies in the interval $tR(\tb{P}_\tb{X},\tb{P}_{\tb{S}|\tb{X}},D_x,D_s)\leq R\leq C(\tb{W}_{\tb{Z}|\tb{Y}})$. For (c), we obtain the maximization of the excess distortion exponent $E^{\mathrm{MIMO}}\left(R,\tb{W}_{\tb{Z}|\tb{Y}}\right)=E_\mathrm{ex}\left(\tb{Q}, \rho^\star, \delta, N_{\mathrm{c}}\right)-\rho^\star R$ over $\delta$ and $\rho$ successively, and solve 
		\begin{align}
		\frac{\partial \widetilde{E}^\mathrm{G}(\frac{R}{t},P_X,P_{S|X})+E^{\mathrm{MIMO}}\left(R,\tb{W}_{\tb{Z}|\tb{Y}}\right)}{\partial R}=0\notag
	\end{align}
	\end{proof}
	Herein, we analyze the properties of the semantic-ware JSCC excess distortion exponent. Due to the matrix essence, we evaluate the expectations on coefficient matrices $\tb{H}$ via the generalized hypergeometry function. Moreover, the partial derivatives are derived in order to solve the optimization problem on JSCC exponent. Note that the statement of exponent is formulated in the form of an optimization problem over coding rate $R$. The solution explicitly presents the optimal design of JSCC scheme for semantic-aware MIMO systems in error exponent sense. 
	\section{Numerical Results}\label{Sec5}
	
	In this section, theoretical bounds of excess distortion exponent in semantic-aware MIMO communication systems are presented. From the simulations, we verify the convexity and show the key quantities like rate-distortion function and ergodic capacity. Then, we explore the influences of MIMO communication systems on semantic information reconstructions, such as coherence time, exponential correlation coefficient and the numbers of antennas. Finally, we also discuss how the optimal code rate $R^\star$ behaviors in terms of different transmission rate $t$.
	\subsection{Experiments Setups}
	We consider the above semantic-aware MIMO communication system in Corollary. \ref{Cor4}, which consists of a joint Gaussian source pair $(\tb{S}, \tb{X})$ with $\bm{\Sigma}_S=3\tb{I}_\ell$ and $\bm{\Sigma}_X=4\tb{I}_q$. Moreover, we assume the same number of transmit and receive antennas $n_T=n_R$, and the channel state matrix $\tb{H}\tb{H}^H\sim\mathcal{Q}(\tb{I}_{n_T},\tb{G}_T,\tb{G}_R)$ (given by \cite{winters1987capacity}), in which we adopt exponential correlation matrices $\tb{G}_T=\{\alpha_T^{|i-j|}\},\tb{G}_R=\{\alpha_R^{|i-j|}\}$ (Simply assuming $\alpha_T=\alpha_R=\alpha\in[0,1)$) to model the spatial correlation. Besides, the signal-to-noise ratio can be calculated by $\mathrm{SNR}=\frac{\mathcal{P}}{N_w}$.
	\subsection{Convexity of JSCC Exponent over Code Rate $R$ for Semantic-Aware MIMO Systems}
	\begin{figure} 
		\label{Sample}
			\centering
			\includegraphics[width=0.6\textwidth]{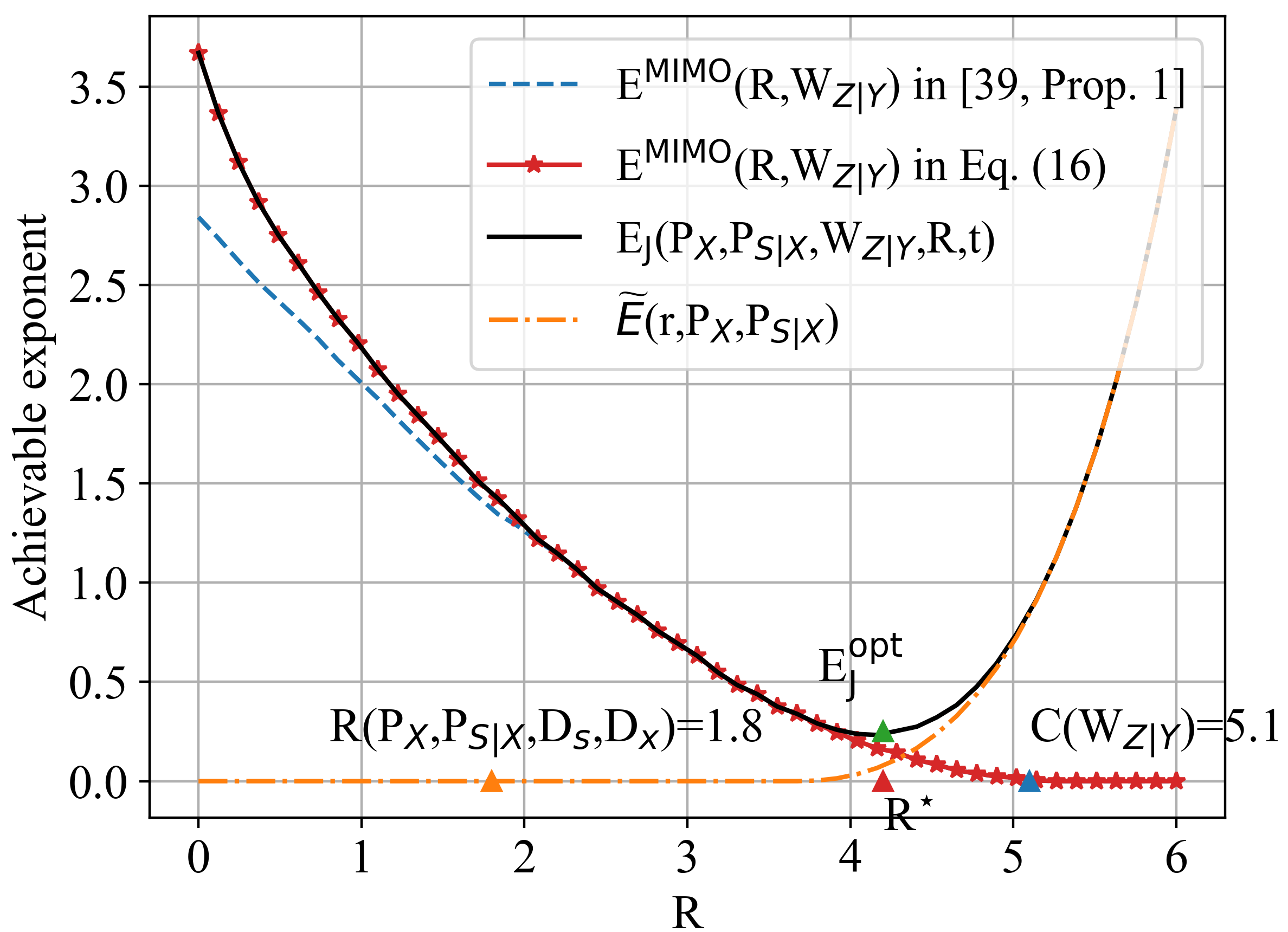}
			\caption{An illustration of the JSCC excess distortion exponent with some key quantities}
	\end{figure}
	In Fig. 3, the JSCC excess distortion exponent for a MIMO system $E_J\left(\tb{P}_\tb{X},\tb{P}_{\tb{S}|\tb{X}},\tb{W}_{\tb{Z}|\tb{Y}},R,t\right)$ is plotted against the code rate $R$. The distortions $D_s=2$, $D_x=1$, coherence time $N_c=1$, $n_T=n_R=3$, correlation coefficient $\alpha=0.3$, SNR$=15$, and the transmission rate $t=2$ symbol/channel use. Note that the rate-distortion function characterizing the semantic information $R(\tb{P}_\tb{X},\tb{P}_{\tb{S}|\tb{X}},D_s,D_x)=1.8$ and the ergodic capacity $C(\tb{W}_{\tb{Z}|\tb{Y}})=5.1$, which is marked in this figure. The solid line represents JSCC exponent function, which consists of source exponent in the dashed dotted line and MIMO channel exponent given in the expurgated bound Eq. \eqref{exb} marked by the star labels. In comparison, we also present the Gallager's random coding bound of MIMO channel \cite{shin2009gallager} in dashed line to verify its suboptimality, since the 'bad' JSCC codewords are expurgated for a better error probability. Due to the convexity of the JSCC exponent, we focus on the minimum of $E_J$, which is labeled by $E_{J}^\mathrm{opt}$ and the optimal JSCC code rate $R^\star$.
	
	
	\subsection{Optimal JSCC Excess Distortion Exponent for Semantic-Aware MIMO Systems}
	In this subsection, we provide plots on the exponential behaviors of the optimal JSCC excess distortion probability for semantic-aware MIMO systems. We investigate how the source and channel key quantities influence the best exponent performance.
	\subsubsection*{Excess Distortion Exponent against Semantic Distortions}	
	\begin{figure}[tbp]
		\centering
		\vspace{-0.35cm}
		\setlength{\abovecaptionskip}{-2pt}
		\subfigure[]{
			\label{Distortion}
			\includegraphics[width=0.45\linewidth]{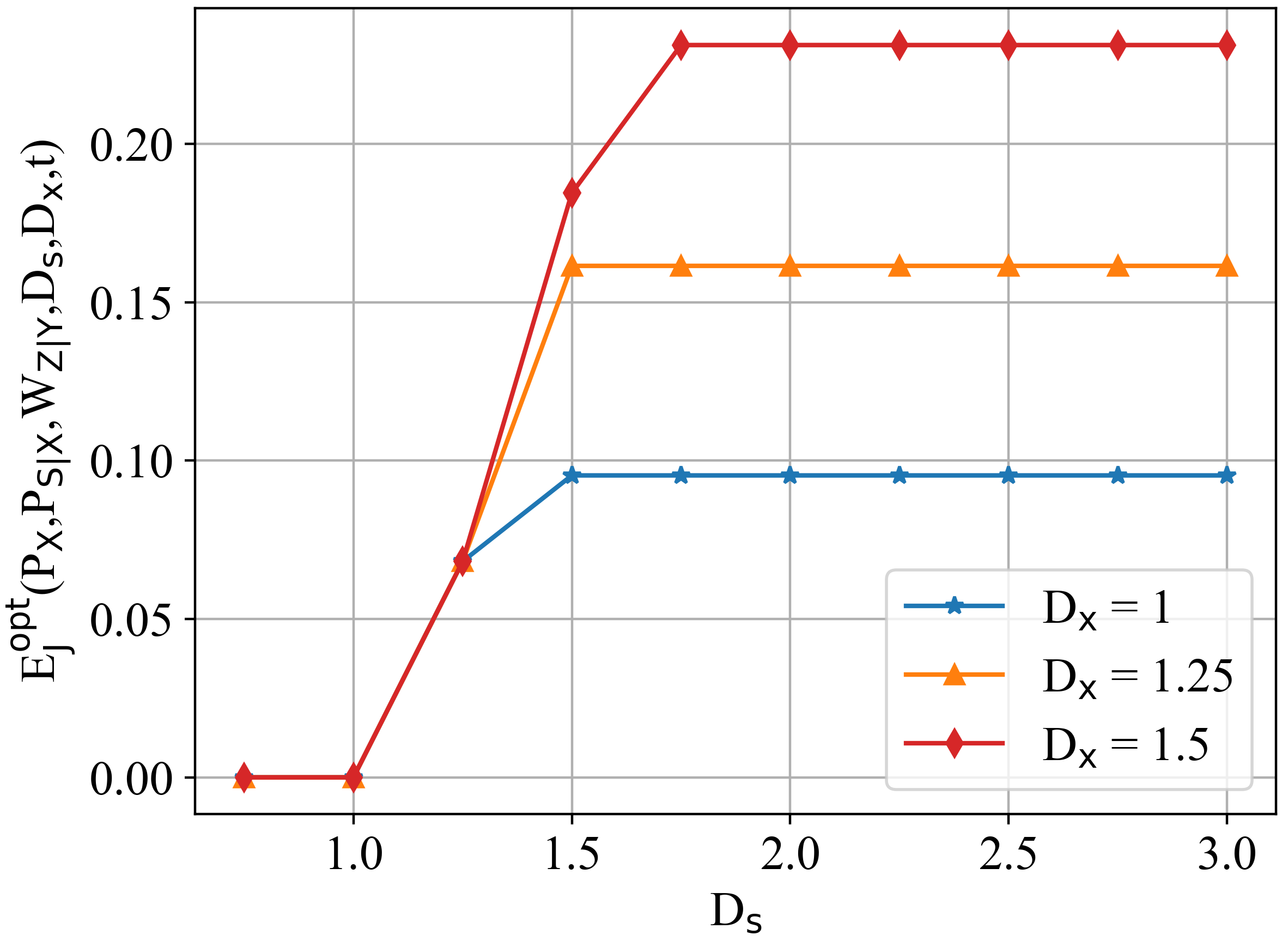}}
		\qquad
		\subfigure[]{
			\label{Attena}
			\includegraphics[width=0.45\linewidth]{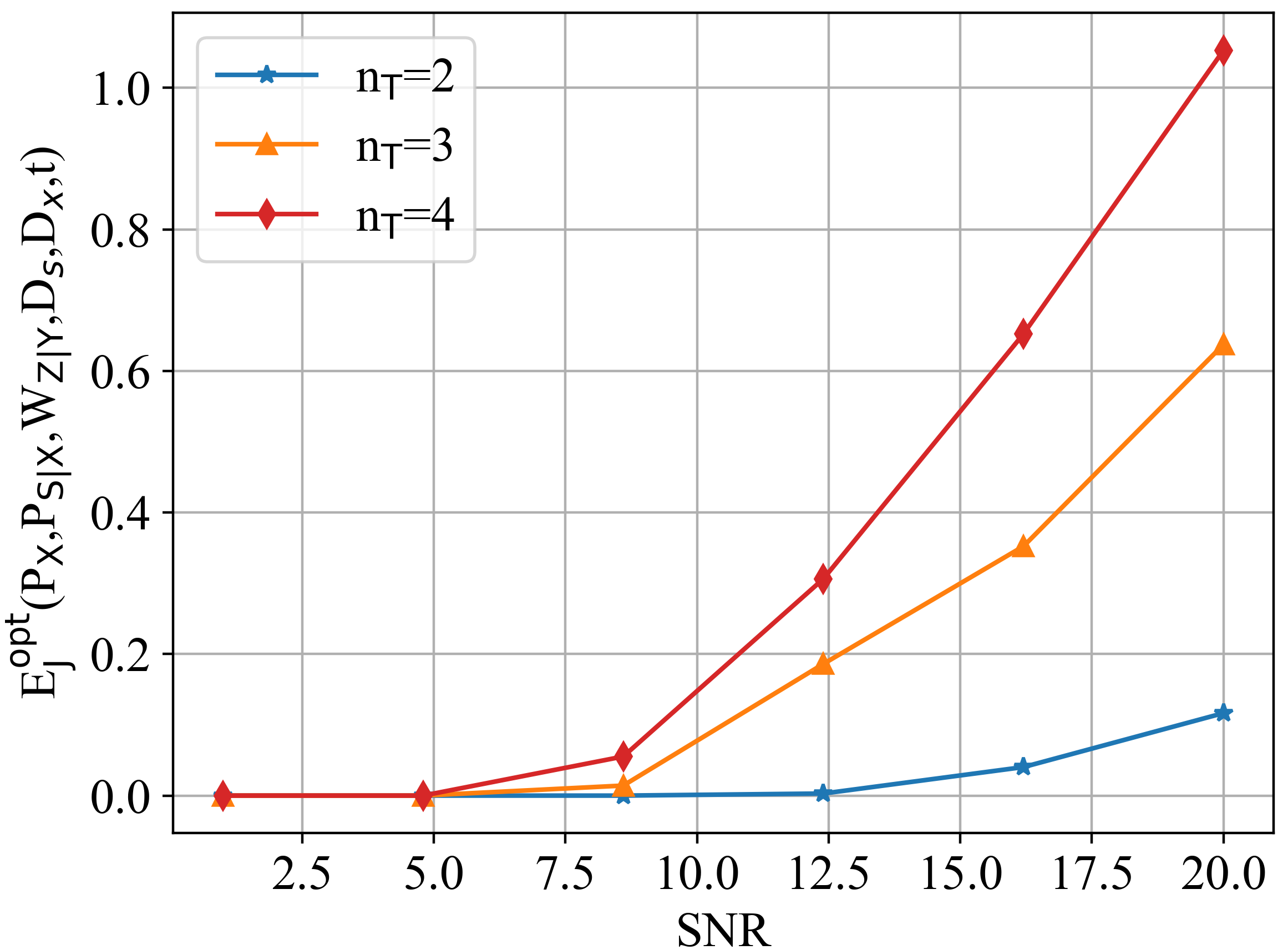}}
		
		\subfigure[]{
			\label{Coherence}
			\includegraphics[width=0.45\linewidth]{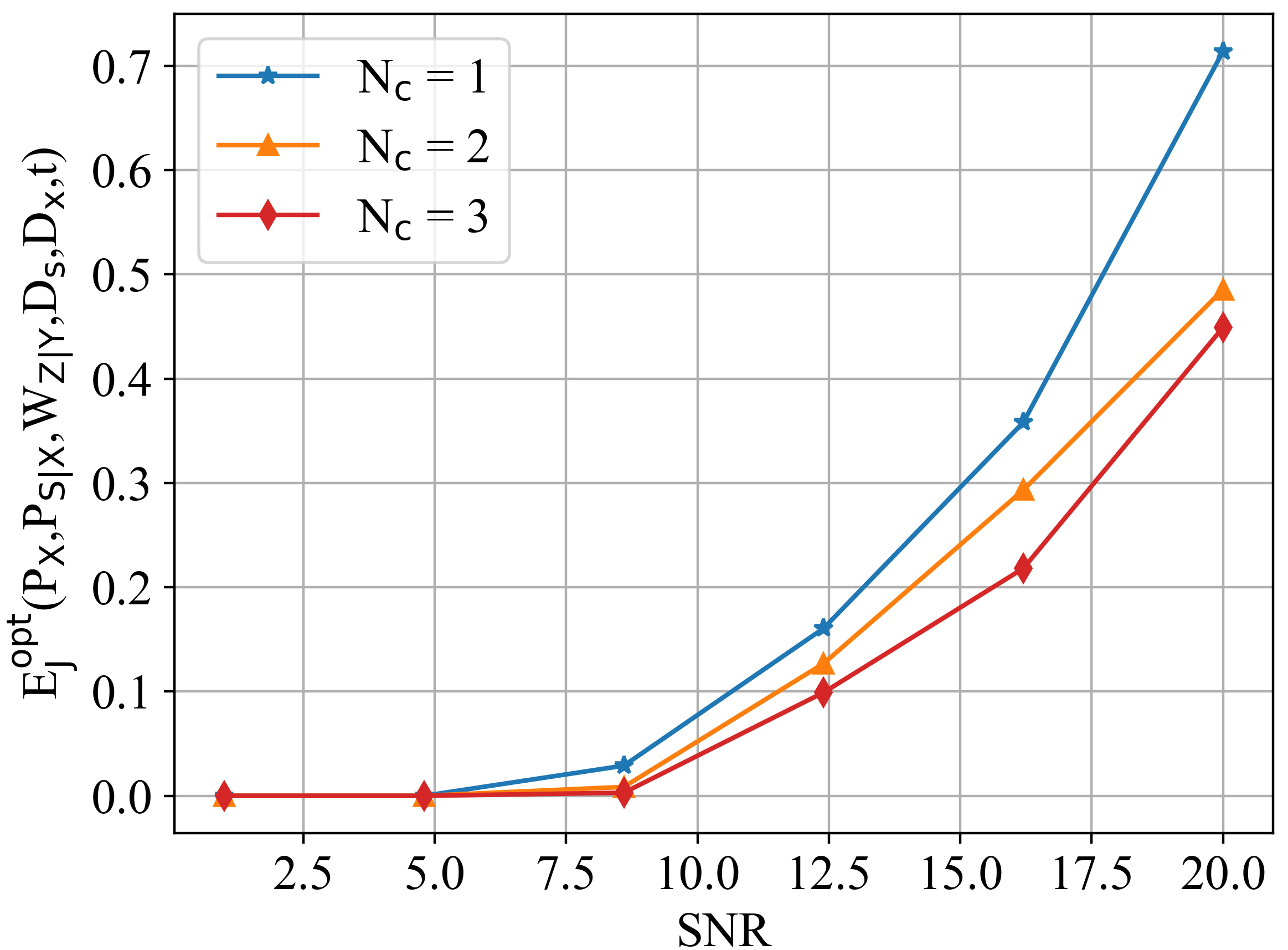}}
		\qquad
		\subfigure[]{
			\label{Coeff}
			\includegraphics[width=0.45\linewidth]{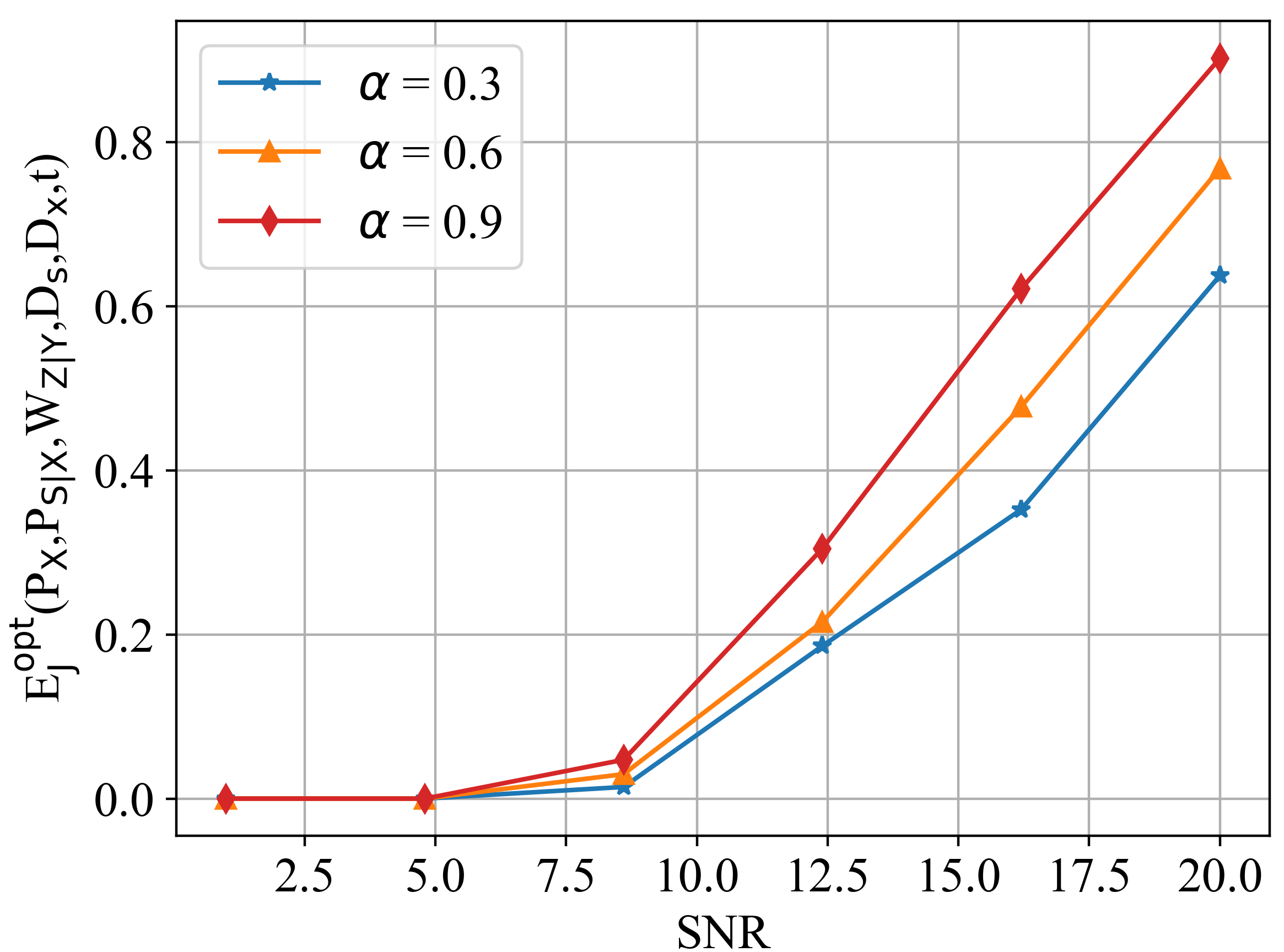}}
		\caption{Excess distortion exponent in terms of: (a) observed distortions, (b) number of antennas, (c) coherence time, (d) correlation coefficient}
	\end{figure}
	Based on the depicted system,
	the excess distortion exponents against semantic distortion are plotted in Fig. \ref{Distortion} in terms of different observed distortions, which intuitively illustrates the tradeoff between $D_s$ and $D_x$. The star, triangle and diamond lines represent the optimal performance with $D_x=1,1.25,1.5$, respectively. Obviously both the increase of $D_s$ and $D_x$ leads to the increase of the exponents, but the curves remain invariant when the semantic distortion becomes inactive, since the observed constraint is more demanding. Under the aforementioned experiment setups, the achievable optimal JSCC excess distortion exponent attains its limit around 0.24 when $D_x=1.5$.
	\subsubsection*{Excess Distortion Exponent against SNR}
We first compare the optimal excess distortion exponent in terms of MIMO systems with different numbers of antennas, namely $n_T=n_R=2,3,4$. From Fig. 4(b), the exponent increases with the number of antennas dramatically. Specifically, under a higher SNR environment, a semantic-aware MIMO system with a $4\times4$ array obtains $E_J^\mathrm{opt}\geq 1$, while a $2\times2$ system only has $1/10$ performance on the exponential probability. This trend demonstrates the compatibility of semantic-aware communication systems and the massive MIMO techniques, with an even larger $n_T$.
	
	In Fig. \ref{Coherence}, we explore how the coherence time $N_c$ in MIMO system affects the exponent. The star, triangle and diamond curves stand for the optimal exponent ranges from 1 to 3, respectively. It can be observed though the longer coherence time results in longer block length, the optimal exponent decreases with $N_c$ at arbitrary SNR. In Fig. \ref{Coeff}, the plot shows the performance in terms of exponential coefficient with $\alpha=0.3,0.6,0.9$, respectively. The optimal exponent increases with $\alpha$, since the larger $\alpha$ means the better channel transmission. 
	\subsection{Achievable Optimal Code Rate $R^\star$ in terms of MIMO Key Quantities}
		\begin{figure}[tbp] 
		\captionsetup[subfigure]{margin=120pt} 
		\begin{minipage}[t]{0.49\linewidth}
			\subfigure[]{
				\label{6a}
				\centering
				\includegraphics[width=1\textwidth]{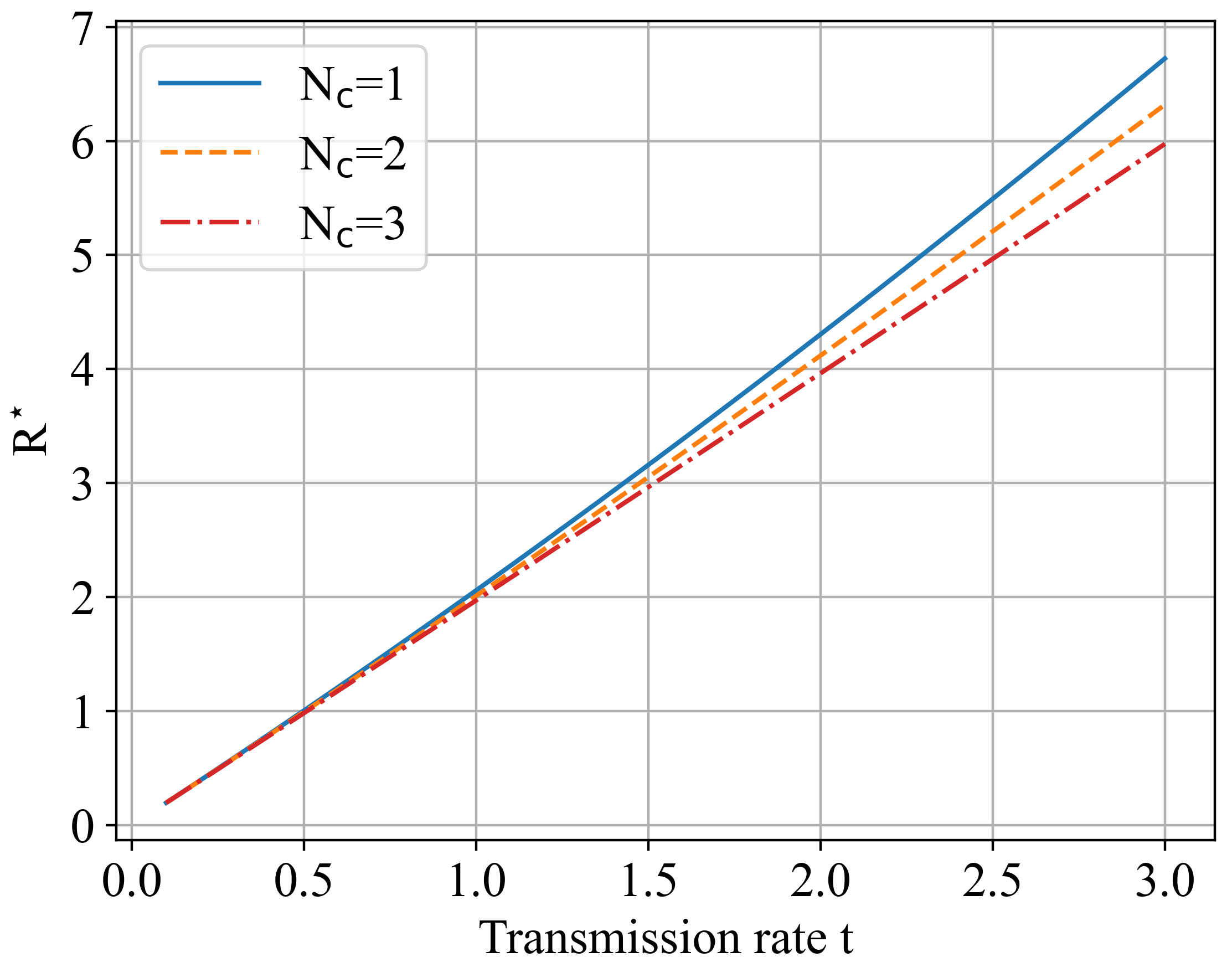}}
		\end{minipage}
		\hspace{0.1in}
		\begin{minipage}[t]{0.49\linewidth}
			\centering
			\subfigure[]{
				\label{6b}
				\includegraphics[width=1\textwidth]{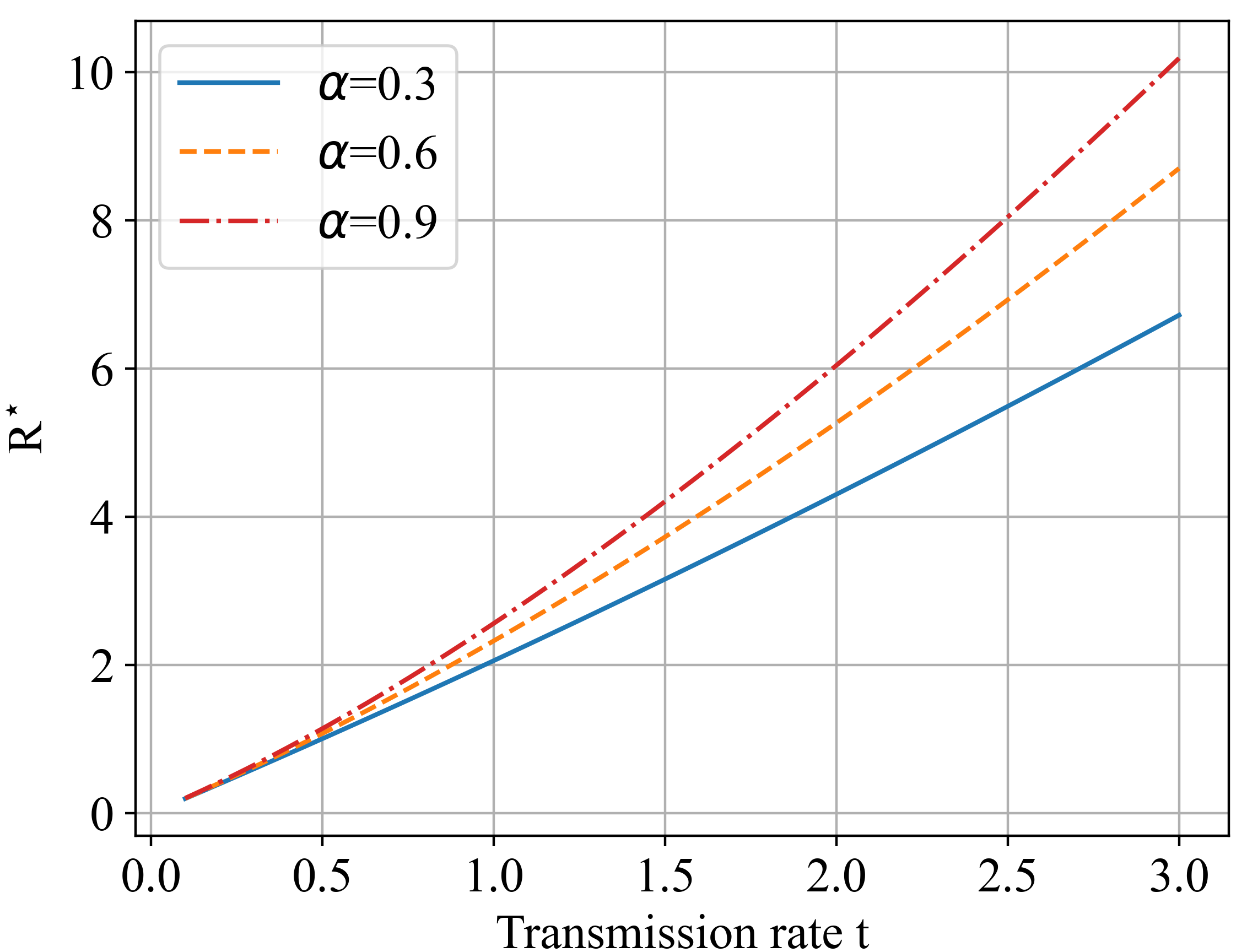}
			}
		\end{minipage}
		\caption{Different JSCC Schemes against Optimal Code Rate $R^\star$ in terms of (a) different coherence time $N_c$ (b) different correlation coefficient}\label{6}
	\end{figure}
	In this subsection, we investigate the optimal JSCC code design for the semantic-aware MIMO systems in error exponent sense. Fixing $D_x = 1.5$, $D_s=2$ and $\alpha=0.3$, the ratio $t=\frac{k}{n}$ is plotted against the optimal code rate $R^\star$ in Fig. \ref{6a}. The transmission rate $t$ increases with the optimal code rate and turns a slightly decrease with coherence time $N_c$. In Fig. \ref{6b}, the exponential correlated coefficient $\alpha$ shows a positive effect on transmission rate $t$. Note that given $t$ and $R^\star$, we obtain the optimal JSCC scheme for the semantic-aware MIMO system.
	
	\section{Conclusion}
	In this paper, we obtained upper and lower bounds of JSCC excess distortion exponent in semantic-aware communication systems. We concluded that the participation of semantic source enlarges the JSCC excess distortion exponent. In a semantic-aware MIMO system, we presented an achievable bound for JSCC excess distortion exponent, which extends the conclusion to a practical communication scenario. As a result, based on the achievability bound, we solved the optimization problem in a simple case and discussed the design of the optimal JSCC scheme in terms of the code rate and transmission rate. Finally, from the numerical results, we show the tradeoff between the two distortions and demonstrated that more antennas and larger correlation efficient lead to a better JSCC excess distortion exponent, while longer coherence time reduces its performance.

	In future works, a tight converse bound for JSCC exponent in semantic-aware MIMO communication system will be considered. Furthermore, the analysis for finite block length case will be also studied for designing practical coding schemes with advantages over the conventional non-semantic-aware communication systems. 


	\appendices
		\section{PROOF OF THEOREM \ref{Maintheorem}}\label{theorem1}
		\begin{lemma}[JSCC Theorem \cite{1998Joint} in Semantic-aware Scenario]\label{JSCCTH}
			Given a memoryless source pair $(S,X)$ and a memoryless channel $W_{Z|Y}$, where $tR(P_X,P_{S|X},D_s,D_x)> C(W_{Z|Y})$, then for any $\left(n,k,D_s,D_x\right)$  JSC code we have
			\begin{align}
				\lim_{n\rightarrow\infty}P_{J}\left(P_{S|X},P_X,W_{Z|Y},n,k\right)=1.
			\end{align}
		\end{lemma}
		\begin{lemma}[c.f.{\cite[Theorem~2]{Shannon_1967}} {\cite[p.~175]{Csiszar_1981}}]\label{lemma2}
			We assume a channel $W$ of capacity $C(W)$ and an $n$-length list code $(\varphi,\psi)$ has a range of decoded sequences of cadinality $l$, which we call the list size.  An erroneous event of list codes occurs when a true message is not on the decoding list. Let $\underline{p_e}(n,R,L)$ denote the minimal error probability, and $\overline{p_e}(n,R,L)$ be the maximal error probability for such an $n$-length list code with rate $R$ and list size $L$. Then for positive $\tilde{\epsilon}_1(n)\rightarrow0$ and $\tilde{\epsilon}_2(n)\rightarrow0$ as $n\rightarrow\infty$, and any $R>C(W)+L$ where $L=\frac{1}{n}\log l$, we have
			\begin{align}
				&\underline{p_e}(n,R,L)\geq\exp\left\{-n\left(E_{\mathrm{sp}}(R-L,W)+\tilde{\epsilon}_1(n)\right)\right\}\label{sp}.\\	&\overline{p_e}(n,R,L)\leq\exp\left\{-n\left(E_{\mathrm{ex}}(R-L,W)-\tilde{\epsilon}_2(n)\right)\right\}\label{rc}.
			\end{align}
		\end{lemma}
		Lemma \ref{lemma2} is a channel coding theorem based on list codes, which characterizes the error probability via sphere-packing bound and expurgated bound. Note that expurgated bound is a refined version of random coding bound, which drops the bad codewords beyond the Bhattacharya distance \cite[Chap~7]{Gallager_1968}.

		The outline of the proof is the following. We start by presenting a strong converse for JSCC coding theorem under semantic-aware communications. To obtain the excess distortion exponent, it is reasonable to model the non-trivial source-channel pairs of $tR(P_X,P_{S|X},D_s,D_x)\leq C(W_{Z|Y})$, which is given in Lemma \ref{JSCCTH}. Here, we focus on investigating the optimal JSCC scheme of code rate within this interval. Then, as mentioned above, the excess distortion probability, which refers to the ratio of over-distorted sequences to overall sequences, can be bounded by computing the sizes of typical sets of sources and channel. Given the source pair $(s^k,x^k)$, the excess distortion probability from the over-distorted codewords is bounded in this appendix via Lemma \ref{lemma2}. Moreover, the average numbers of the corrupted sequences from the semantic and observable sources are stated as well, by investigating the typical sequences $x^k$ and the conditional typical sequences $s^k$. Finally, by combining these results from source with channel parts and minimizing the sum in terms of code rate $R$, we obtain final results.
		
		According to Lemma \ref{JSCCTH}, we investigate the excess distortion performance of a group of JSCC schemes with code rate $tR(P_X,P_{S|X},D_s,D_x)\leq R\leq C(W_{Z|Y})$.  Now we recall the overall excess distortion probability defined in Eq. \eqref{statement}, and rewrite
		\begin{align}
			\mathbb{P}\left\{\mathcal{E}\right\}
			=&\sum_{x^k\in\mathcal{X}^k}P_{X^k}(x^k)\sum_{s^k\in\mathcal{S}^k}P_{S^k|X^k}(s^k|x^k)p_c(s^k,x^k)\label{overall},
		\end{align}
		where we use $p_c(s^k,x^k)\triangleq\sum_{z^n\in\mathcal{E}(s^k,x^k)}P_{Z^n|Y^n}(z^n|\varphi^n(x^k))$.
		For the excess distortion probability $p_c(s^k,x^k)$, let $\mathcal{T}_Q$ denote the typical set of sequences $x^k\in\mathcal{X}^k$ of type $Q_X$,  $\mathcal{T}_U$ be the joint typical set of sequences $(s^k,x^k)\in\mathcal{S}^k\times\mathcal{X}^k$, and $\mathcal{T}_{U}(x^k)=\left\{s^k:(s^k,x^k)\in\mathcal{T}_U\right\}$ is the conditional typical set of $s^k$ for a given $x^k$, in which the conditional empirical distribution is $U_{S|X}$. Moreover, the conditional typical set $\mathcal{T}_V(z^n)=\{(s^k,x^k):(s^k,x^k,z^k)\in\mathcal{T}_V\}$ based on the joint typical set $\mathcal{T}_V$ composed of sequence tuple $(s^k,x^k,z^k)$. Note that the size of the conditional typical set is 
		\begin{align}
			|\mathcal{T}_V(z^n)|
			\leq&\exp\left\{kH(S^\star,X^\star|Z^\star)\right\}
			=\exp\left\{k\left(H(S^\star,X^\star)-I(S^\star, X^\star;Z^\star)\right)\right\}\notag\\
			\leq&\exp\left\{k\left(H(Q_X,U_{S|X})-R(Q_X,U_{S|X},D_s,D_x)\right)\right\}\notag,
		\end{align} 
		where $S^\star$, $X^\star$ and $Z^\star$ are three arbitrary auxiliary random variables characterizing the joint distribution $P_{S^\star X^\star Z^\star}$, which is a possible joint type of sequences $x^k\in\mathcal{T}_{Q}$, $s^k\in\mathcal{T}_U(x^k)$ and $z^n\in\mathcal{Z}^n$ within the distortion constraints. Next, the number of all possible joint types is upper bounded by $(k+1)^{|\mathcal{S}||\mathcal{X}||\mathcal{Z}|}$ via the type counting lemma. Hence the list size can be bounded by
		\begin{align}
			l&\leq(k+1)^{|\mathcal{S}||\mathcal{X}||\mathcal{Z}|}\exp\left\{k\left(H(Q_X,U_{S|X})\right.\right.
			\left.\left.-R(Q_X,U_{S|X},D_s,D_x)\right)\right\}\notag\\
			&=\exp\left\{k\left(H(Q_X,U_{S|X}\right)\right.\left.\left.-R(Q_X,U_{S|X},D_s,D_x)+\hat{\epsilon}_1(k)\right)\right\},\notag
		\end{align}
		where 
		$
		k = nt$,  $\hat{\epsilon}_1(k)=\frac{1}{k}\log(k+1)^{|\mathcal{S}||\mathcal{X}||\mathcal{Z}|}.
		$
		Note that 
		\begin{align}
			\lim_{k\rightarrow\infty}\hat{\epsilon}_1(k)
			= \lim_{k\rightarrow\infty}\frac{1}{k}\log(k+1)^{|\mathcal{S}||\mathcal{X}||\mathcal{Z}|}
			\overset{(a)}{=}|\mathcal{S}||\mathcal{X}||\mathcal{Z}| \lim_{k\rightarrow\infty}\frac{1}{k}\log(k+1)
			=0,\label{ep}
		\end{align}
		where equality (a) is because the product of alphabet cardinalities, $|\mathcal{S}||\mathcal{X}||\mathcal{Z}|$, is a finite constant. Thus, the parameter $L$ is
		\begin{align}
			L&=\frac{1}{n}\log l\leq t\left(H(Q_X,U_{S|X})-R(Q_X,U_{S|X},D_s,D_x)+\hat{\epsilon}_1(k)\right)\label{L}.
		\end{align}Since the size of the message set  satisfies
		\begin{align}
			\hspace{-0.22cm}\exp\{nR\}\leq\exp\{kH(P_{S,X})\}=\exp\{ntH(Q_X,U_{S|X})\}\label{R},
		\end{align}
		by substituting Eq. \eqref{L} and Eq. \eqref{R} into Eq. \eqref{sp}, we have
		\begin{align}
			p_{c}(s^k,x^k)
			&=\exp\left\{-n\left(E_{\mathrm{sp}}(tR(Q_X,U_{S|X},D_s,D_x)-t\hat{\epsilon}_1(k),W_{Z|Y})+\tilde{\epsilon}_1(n)\right)\right\}\notag\\
			&\overset{(b)}{\geq}\exp\left\{-n\left(E_{\mathrm{sp}}(tR(Q_X,U_{S|X},D_s,D_x),W_{Z|Y})+\epsilon_1(n)\right)\right\}\label{final},
		\end{align} 
		where the inequality $(b)$ holds since  $E_{\mathrm{sp}}(R,W)$ is a non-increasing function in $R$, and $\epsilon_1(n)=t\hat{\epsilon}_1(n)+\tilde{\epsilon}_1(n)\rightarrow0$ as $n\rightarrow\infty$ .
		
		Now given an observed sequence $x^k$, we define 
		$
		p_a(x^k)=\sum_{s^k\in\mathcal{S}^k}P_{S^k|X^k}(s^k|x^k)p_c(s^k,x^k),
		$ and
		\begin{align}
			p_a(x^k)=&		\sum_{U_{S|X}\in\mathcal{U}}\sum_{s^k\in\mathcal{T}_U(x^k)}P_{S^k|X^k}(s^k|x^k)p_c(s^k,x^k)\notag\\
			\overset{(c)}{=}&\sum_{U_{S|X}\in\mathcal{U}}\sum_{s^k\in\mathcal{T}_U(x^k)}\prod_{(a,b)\in(\mathcal{S}\times\mathcal{X})}P_{S|X}(a|b)^{N((a,b)|(s^k,x^k))}p_c(s^k,x^k)\notag\\
			=&\sum_{U_{S|X}\in\mathcal{U}}\left\vert \mathcal{T}_U\left(x^k\right)\right\vert
			\exp\left\{-k\mathbb{E}_{Q_X\times U_{S|X}}\left[-\log P_{S|X}(S|X)\right]\right\}p_c(s^k,x^k)\notag\\
			\overset{(d)}{=}&\sum_{U_{S|X}\in\mathcal{U}}\left\vert \mathcal{T}_U\left(x^k\right)\right\vert\exp\left\{-kH(U_{S|X}\vert Q_{X})\right\}\exp\left\{-kD\left(\left.U_{S|X}\right\Vert \left.P_{S|X}\right\vert Q_X\right)\right\}p_c(s^k,x^k)\notag\\
			\overset{(e)}{\geq}&\sum_{U_{S|X}\in\mathcal{U}}(k+1)^{-|\mathcal{S}||\mathcal{X}|}\exp\left\{-kD\left(\left.U_{S|X}\right\Vert \left.P_{S|X}\right\vert Q_X\right)\right\}p_c(s^k,x^k)\label{c}.
		\end{align}
		where $\mathcal{Q}$ denotes the set of all types $Q_X$ and $\mathcal{U}$ denotes the set of conditional types $U_{S|X}$. In $(c)$ we apply the empirical count function $N((a,b)|(s^k,x^k))$ on the conditional probability, in $(d)$ we use the definition on the conditional divergence defined in Eq. \eqref{Cond}, and in $(e)$ the following result is used (\cite[Lemma~2.3]{Csiszar_1981}):
		\begin{align}
			(k+1)^{-|\mathcal{S}||\mathcal{X}|}\leq\left\vert\mathcal{T}_U\left(x^k\right)\right\vert\exp\left\{-nH(U_{S|X}\vert Q_{X})\right\}\leq 1.\notag
		\end{align}
		To characterize the possible conditional types $U_{S|X}$, given the observable sequences $x^k\in\mathcal{T}_Q$, we use the conclusion that the code rate is upper bounded by rate-distortion function $tR(Q_X,U_{S|X},D_s,D_x)\geq R$. Thus all possible $U_{S|X}$ should be restricted in the following set:
		\begin{align}
			\mathcal{U}\triangleq\left\{U_{S|X}\in\mathcal{C}(\mathcal{X}\rightarrow\mathcal{S}):R(Q_X,U_{S|X},D_s,D_x)\geq r\right\}.\label{u}
		\end{align}
		Following Eq. \eqref{overall}, the overall excess distortion probability can be stated in Eq. \eqref{over},
		\begin{align}
			\mathbb{P}\left\{\mathcal{E}\right\}
			=&\sum_{Q_X\in\mathcal{Q}}\sum_{x^k\in\mathcal{T}_{Q}}P_{X^k}(x^k)p_a(x^k)\notag\\
			\geq&\exp\left\{-k\left(\min_{Q_X\in\mathcal{Q}}\left(D(Q_X||P_X)+\min_{U_{S|X}\in\mathcal{U}}D(U||P_{S|X}|Q_X)+\epsilon_3(k)\right)\right)\right\}p_c(s^k,x^k)\notag\\
			=&\exp\left\{-k\left(\widetilde{E}(r,P_X,P_{S|X})+\epsilon_3(k)\right)\right\}\exp\left\{-n\left(E_{\mathrm{sp}}(R,W_{Z|Y})+\epsilon_{1}(n)\right)\right\}\label{over},
		\end{align}
		where $\epsilon_{1}(n)$ is given in Eq. \eqref{final} and
		\begin{align}
			\epsilon_3(k) &= -\frac{1}{k}\log\left(|\mathcal{Q}||\mathcal{U}|(k+1)^{-|\mathcal{S}||\mathcal{X}|-|\mathcal{X}|}\right)\overset{(f)}{\leq}\frac{1}{k}\log\left((k+1)^{|\mathcal{X}|}\right).\notag
		\end{align}
		Inequality $(f)$ holds since the number of all possible joint types of $x^k$ and $s^k$ is upper bounded by\footnote{More specifically, $\mathcal{U}\subseteq\mathcal{U}^\star=\{U_{S|X}:U_{S|X}\in\mathcal{C}(\mathcal{X}\rightarrow\mathcal{S})\}$, and the number of all joint types is bounded by type counting lemma \cite[Lemma~2.2]{Csiszar_1981} as $|\mathcal{Q}||\mathcal{U}|\leq|\mathcal{Q}||\mathcal{U}^\star|\leq(k+1)^{|S||\mathcal{X}|}$.} $|\mathcal{Q}||\mathcal{U}|\leq(k+1)^{|\mathcal{S}||\mathcal{X}|}$, hence $\epsilon_3(k)\rightarrow0$ as $k\rightarrow\infty$. 
		Combining Eq. \eqref{over} with the definition of $\widetilde{E}(r,P_X,P_{S|X})$ in Eq. \eqref{sourcecoding},
		we state the upper bound on the exact excess distortion exponent as
		\begin{align}
			\widetilde{E}_J(R)\triangleq\liminf_{n\rightarrow\infty}\left[-\frac{1}{n}\log \mathbb{P}\left\{\mathcal{E}\right\}\right]\leq \frac{k}{n}\widetilde{E}(r,P_X,P_{S|X})+E_{\mathrm{sp}}(R,W_{Z|Y})\label{1}.
		\end{align} 
		With the achievability Eq. \eqref{rc} in Lemma \ref{lemma2}, we can get the lower bound similarly on $\widetilde{E}_J(R)$
		\begin{align}
			\widetilde{E}_J(R)\geq \frac{k}{n}\widetilde{E}(r,P_X,P_{S|X})+E_{\mathrm{ex}}(R,W_{Z|Y})\label{4}.
		\end{align} 
		By substituting $t=\frac{k}{n}$ and $r=\frac{R}{t}$ in Eq. \eqref{1} and Eq. \eqref{4}, respectively, we finally obtain upper and lower bounds of excess distortion exponent in Eq. \eqref{Upp} and Eq. \eqref{Low}, respectively. Note that the excess distortion probability is stated as a function of $n,k$ but the exponent of the optimal JSCC only concerns the transmission rate $t$ symbol/channel use.
		
		\section{PROOF OF Theorem \ref{Gaussian}}\label{Gauss}
		\begin{lemma}\cite[Appendix~B]{Taricco08}\label{Matint}
			For every pair of Hermitian positive definite matrices $\mathbf{A} \in \mathbb{C}^{m \times m}, \mathbf{C} \in \mathbb{C}^{n \times n}$, and denote $\mathrm{etr}(\cdot)=\exp\{\mathrm{tr}(\cdot)\}$, then for any matrices $\mathbf{B}, \mathbf{D} \in \mathbb{C}^{m \times n}$, we have:
			$$
			\begin{aligned}
				\int_{\mathbb{C}^{n \times m}} \operatorname{etr}\left(-\pi\left(\mathbf{A} \mathbf{U}^{\mathbf{H}} \mathbf{C} \mathbf{U}+\mathbf{B}^{\mathbf{H}} \mathbf{U}+\mathbf{U}^{\mathbf{H}} \mathbf{D}\right)\right) d \mathbf{U} &=\operatorname{det}\left(\mathbf{A}^{\top} \otimes \mathbf{C}\right)^{-1} \operatorname{etr}\left(\pi \mathbf{A}^{-1} \mathbf{B C}^{-1} \mathbf{D}^{\mathbf{H}}\right) \\
				&=\operatorname{det}(\mathbf{A})^{-n} \operatorname{det}(\mathbf{C})^{-m} \operatorname{etr}\left(\pi \mathbf{A}^{-1} \mathbf{B C}^{-1} \mathbf{D}^{\mathbf{H}}\right)
			\end{aligned}
			$$
		\end{lemma}
	
		This section shows how to derive Eq. \eqref{MMo} from Eq. \eqref{Low} under a semantic-aware MIMO communication system. Specifically, we derive the explicit forms of the source excess distortion exponent $\widetilde{E}(r,\tb{P}_\tb{X},\tb{P}_{\tb{S}|\tb{X}})$, and the expurgated random coding bound $E_{\mathrm{ex}}(R,\tb{W}_{\tb{Z}|\tb{Y}})$ on channel excess distortion exponent under MIMO communication systems. The basic idea of the proof is, for the excess distortion exponent of source pairs, we rewrite the K-L divergence according to the generalized rate-distortion function into a computable optimization problem. For the MIMO channel expurgated random coding exponent, we utilize the equivalence between Csiszar's form and Gallager's form by Fenchel duality. 
		%

		We start from a jointly Gaussian distributed vector pair $(\tb{S}=\tb{h}\tb{X}+\tb{N},\tb{X})$, and the reconstructed source vectors $\hat{\tb{S}}$, $\hat{\tb{X}}$. The minimum of Eq. \eqref{sourcecoding} is presented with two distributions $\tb{Q}_\tb{X}$ and $\tb{U}_{\tb{S}|\tb{X}}$ subject to  $R(\tb{Q}_\tb{X},\tb{U}_{\tb{S}|\tb{X}},D_s,D_x)\geq r$. Note that under the Gaussian vector assumption, the semantic-aware rate-distortion function can be rewritten as
		\begin{align}
			&R(\tb{Q}_\tb{X},\tb{U}_{\tb{S}|\tb{X}},D_s,D_x)=\min\frac{1}{2}\log\left(\frac{\det(\tb{A})}{\det(\bm{\Delta})}\right)\label{a1}\\
			\text{s.t.}\qquad& \bm{O}\prec \bm{\Delta}\preceq\tb{A},\\
			&\mathrm{tr}\left\{\tb{h}\bm{\Delta}\tb{h}^H\right\}\leq D_s-\mathrm{tr}\{\tb{B}\},\\
			&\mathrm{tr}\left\{\bm{\Delta}\right\}\leq D_x,
		\end{align}
	Moreover, in Eq. \eqref{sourcecoding}, the first divergence can be computed as
	\begin{align}
		D(\tb{Q}_\tb{X}||\tb{P}_\tb{X})&=\mathbb{E}_{\tb{Q}_\tb{X}}\left[\log\tb{Q}_\tb{X}-\log\tb{P}_\tb{X}\right]\notag\\
		&=\mathbb{E}_{\tb{Q}_\tb{X}}\left[\frac{1}{2}\log\frac{\det(\bm{\Sigma}_X)}{\det(\tb{A})}+\frac{1}{2}\tb{x}^H\left(\bm{\Sigma}_X^{-1}-\tb{A}^{-1}\right)\tb{x}\right]\notag\\
		&=\frac{1}{2}\log\frac{\det(\bm{\Sigma}_X)}{\det(\tb{A})}-q+\mathrm{tr}\left\{\bm{\Sigma}_X^{-1}\tb{A}\right\}\label{41}
	\end{align}
where the auxiliary multivariate Gaussian distribution $\tb{Q}_\tb{X}(\tb{x})=\frac{1}{(2\pi)^{\frac{q}{2}}\det^\frac{1}{2}(\tb{A})}\exp\{\frac{1}{2}\tb{x}^H\tb{A}^{-1}\tb{x}\}$. Similarly, the conditional divergence in Eq. \eqref{sourcecoding} can be stated as
\begin{align}
	D(\tb{U}_{\tb{S}|\tb{X}}||\tb{P}_{\tb{S}|\tb{X}}|\tb{Q}_\tb{X})&=\mathbb{E}_{\tb{Q}_\tb{X}\times\tb{U}_{\tb{S}|\tb{X}}}[\log\tb{U}_{\tb{S}|\tb{X}}-\log\tb{P}_{\tb{S}|\tb{X}}]\notag\\
	&=\frac{1}{2}\log\frac{\det(\bm{\Sigma}_N)}{\det(\tb{B})}+\mathrm{tr}\left\{\left(\bm{\Sigma}_N-\tb{B}^{-1}\right)\tb{h}^H\bm{\Sigma}_X\tb{h}+\bm{\Sigma}_N^{-1}\tb{B}\right\}-\ell\label{42}
\end{align}
where the distribution $\tb{U}_{\tb{S}|\tb{X}}(\tb{s}|\tb{x})=\frac{1}{(2\pi)^{\frac{l}{2}}\det^\frac{1}{2}(\tb{B})}\exp\{\frac{1}{2}\left(\tb{s}-\tb{h}\tb{x}\right)^H\tb{B}^{-1}\left(\tb{s}-\tb{h}\tb{x}\right)\}$. Finally, combining Eq.\ eqref{41} and Eq. \eqref{42} with the constraints from Eq. (37)-(39), we obtain the semantic-aware source excess distortion exponent as Eq. \eqref{Gsource}.

		In Theorem \ref{Maintheorem}, we state the channel exponent in Csiszar's form (Eq. \eqref{channelex}) in light of the simplicity on statement, but it is hard to be computed. Notably, Zhong \cite{Zhong08} combined the Csiszar's exponent with Gallager's reliability function via Fenchel duality and proved the equivalence, in which the later is easy to be extended and analyzed. The reader can turn to \cite[Chapter~4]{Zhong08} \cite[Chapter~7]{Gallager_1968} for more details.
		
		Under the MIMO system, Gallager's random coding bound is given by \cite[Prop~1]{shin2009gallager}. For an expurgated bound, which is derived from a codebook expurgating the bad codewords, and hence performs better than random coding bound at lower code rate, Alfano \cite[Thm~3.2]{Alfano_Chiasserini_Nordio_Zhou_2015} evaluated it under simple assumptions. Herein we present the expurgated bound on error exponent in terms of random matrices. From \cite{Gallager_1968}, the expurgated exponent is stated as
		\begin{small}
			\begin{align}
				E_{\mathrm{ex}}=-\frac{1}{N_c}\ln\int_{\tb{H}} p_{\tb{H}}(\tb{H})\left\{\int_{\tb{Y}\tilde{\tb{Y}}}p_{\tb{Y}\tilde{\tb{Y}}}(\tb{Y}\tilde{\tb{Y}})\exp\left\{\delta\left[\mathrm{tr}\left(\tb{Y}\tb{Y}^H+\tilde{\tb{Y}}\tilde{\tb{Y}}^H\right)-2\mathcal{P}\right]\right\}w(\tb{Y},\tilde{\tb{Y}},\tb{Z})^{\frac{1}{\rho}}d\tb{Y}\tilde{\tb{Y}}\right\}^\rho d\tb{H}\label{ex},
			\end{align}
		\end{small}
		where $w(\tb{Y},\tilde{\tb{Y}},\tb{Z})=\int_{\tb{Z}}\sqrt{p(\tb{Z}|\tb{Y},\tb{H})p(\tb{Z}|\tilde{\tb{Y}},\tb{H})}d\tb{Z}
		$ and $-\ln w(\tb{Y},\tilde{\tb{Y}},\tb{Z})$ is the Bhattacharya distance between channel input matrices $\tb{Y}$ and $\tilde{\tb{Y}}$ while $\tb{Z}$ is the channel output. Next we first process the aforementioned integral with the transition probability \eqref{MIMOtran} of MIMO system as
		\begin{align}
			w(\tb{Y},\tilde{\tb{Y}},\tb{Z})&=(\pi N_w)^{-n_RN_c}\exp\left\{-\frac{1}{2N_w}\mathrm{tr}\left(
			\tb{H}\tb{Y}\tb{Y}^H\tb{H}^H+\tb{H}\tilde{\tb{Y}}\tilde{\tb{Y}}^H\tb{H}^H\right)\right\}\notag\\
			&\times\int_{\tb{Z}}\exp\left\{-\frac{1}{2N_w}\mathrm{tr}\left(2\tb{Z}\tb{Z}^H-\tb{Z}\tb{Y}^H\tb{H}^H-\tb{H}\tb{Y}\tb{Z}^H-\tb{Z}\tilde{\tb{Y}}^H\tb{H}^H-\tb{H}\tilde{\tb{Y}}\tb{Z}^H\right)\right\}d\tb{Z}\notag\\
			&=\exp\left\{-\frac{1}{4N_w}\mathrm{tr}\left(\tb{H}\left(\tb{Y}-\tilde{\tb{Y}}\right)\left(\tb{Y}-\tilde{\tb{Y}}\right)^H\tb{H}^H\right)\right\}\label{l4}
		\end{align}
		where \eqref{l4} follows from Lemma \ref{Matint}. Moreover, by assuming a capacity achieving input distribution on matrix $\hat{\tb{Y}}$, we obtain
		\begin{align}
			&\int_{\tilde{\tb{Y}}}p_{\tilde{\tb{Y}}}(\tilde{\tb{Y}})\exp\left\{\delta\mathrm{tr}\left(\tilde{\tb{Y}}\tilde{\tb{Y}}^H\right)\right\}w(\tb{Y},\tilde{\tb{Y}},\tb{Z})^\frac{1}{\rho}d\tilde{\tb{Y}}\notag\\
			=&\pi^{-n_TN_c}\det(\tb{Q})^{-N_c}
			\int_{\tilde{\tb{Y}}}\exp\left\{\mathrm{tr}\left(\left(\delta\tb{I}_{n_T}-\tb{Q}^{-1}-\frac{1}{4N_w\rho}\tb{H}^H\tb{H}\right)\tilde{\tb{Y}}\tilde{\tb{Y}}^H\right.\right.\notag\\
			&\left.\left.-\frac{1}{4N_w\rho}\tilde{\tb{Y}}^H\tb{H}^H\tb{H}\tb{Y}-\frac{1}{4N_w\rho}\tb{Y}^H\tb{H}^H\tb{H}\tilde{\tb{Y}}\right)\right\}d\tilde{\tb{Y}}\notag\\
			=&\det(\tb{Q})^{-N_c}\det(\tb{A})^{-N_c}\exp\left\{\mathrm{tr}\left(\frac{1}{16N_w^2\rho^2}\tb{A}^{-1}\tb{Y}^H\tb{H}^H\tb{H}\tb{H}^H\tb{H}\tb{Y}\right)\right\}\label{52}.
		\end{align}
		By applying Lemma \ref{Matint} again and $\tb{A}=\delta\tb{I}_{n_T}-\tb{Q}^{-1}-\frac{1}{4N_w\rho}\tb{H}^H\tb{H}$, we achieve equation \eqref{52}. The expectation on input matrix $\tb{Y}$ can be formulated as
		\begin{align}
			&\det(\tb{Q}\tb{A})^{-N_c}\int_{\tb{Y}}p_{\tb{Y}}(\tb{Y})\exp\left\{\delta\mathrm{tr}\left(\tb{Y}\tb{Y}^H-2\mathcal{P}\right)\right\}\exp\left\{\mathrm{tr}\left(\frac{1}{16N_w^2\rho^2}\tb{A}^{-1}\tb{Y}^H\tb{H}^H\tb{H}\tb{H}^H\tb{H}\tb{Y}\right)\right\}d\tb{Y}\notag\\
			&=\det(\tb{Q}\tb{A})^{-N_c}\exp\{-2\delta N_C\mathcal{P}\}\int_{\tb{Y}}p_{\tb{Y}}(\tb{Y})\exp\left\{\mathrm{tr}\left(\frac{\tb{H}^H\tb{H}\tb{A}^{-1}\tb{H}^H\tb{H}}{16N_w^2\rho^2}-\frac{\tb{H}^H\tb{H}}{4N_w\rho}+\delta\right)\tb{Y}\tb{Y}^H\right\}d\tb{Y}\notag\\
			&=\exp\{-2rN_C\mathcal{P}\}\det(\tb{Q}\tb{A})^{-N_c}\det\left(\tb{I}_{n_T}-\tb{Q}\left(\frac{\tb{H}^H\tb{H}\tb{A}^{-1}\tb{H}^H\tb{H}}{16N_w^2\rho^2}-\frac{\tb{H}^H\tb{H}}{4N_w\rho}+\delta\right)\right)^{-N_c}\label{53}
		\end{align}
		Finally substituting \eqref{53} into \eqref{ex} yields the expurgated bound on the error exponent in \eqref{exb}.
		
		In conclusion, we state the source exponent as piecewise function Eq. (16). Finally, combining Eq. (16) with the statement of expurgated random coding bound of channel exponent, we complete the proof of Theorem \ref{Gaussian}.

	\ifCLASSOPTIONcaptionsoff
	\newpage
	\fi


	
	%
	%
	\bibliographystyle{IEEEtran}
	\bibliography{Ref}
	
	%

	
	
	
\end{spacing}

\end{document}

%% file: figure/ChannelModel.tex
	\begin{tikzpicture}[
		outer/.style={draw=gray,dashed,fill=black!1,thick,inner sep=3pt},
		box/.style={draw,thick,minimum width=10mm,minimum height=8mm,text width =20mm,align=flush center},
		box1/.style={draw,fill=gray!30,thick,minimum width=10mm,minimum height=8mm,text width =20mm,align=flush center},
		arrow/.style={draw, single arrow,thick,minimum height=8mm, minimum width=8mm, single arrow head extend=2mm,anchor=west}]
		
		\node (s1) at (0,0) [] {\Large$S$};
		\node (s2) at (2.2,0) [box] {\Large$P_{X|S}$};
		\node (s3) at (5.5,0) [box1] {\Large$\varphi{(X)}$};
		\node (s4) at (8.7,0) [box] {\Large$P_{Z|Y}$};
		\node (s) at (11.1,0) [] {\Large$Z$};
		\node (s5) at (13.1,1) [box1] {\Large$\psi_S{(Z)}$};
		\node (s6) at (13,-1) [box1] {\Large$\psi_X{(Z)}$};
		\draw [thick,shorten >=1mm,shorten <=1mm,->](s1.east)->(s2.west)node [midway,sloped,above,align=center] {};
		\draw [thick,shorten >=1mm,shorten <=1mm,->](s2.east)->(s3.west)node [midway,sloped,above,align=center] {\Large$X$};
		\draw [thick,shorten >=1mm,shorten <=1mm,->](s3.east)->(s4.west)node [midway,sloped,above,align=center] {\Large$Y$};
		\draw [thick,shorten >=1mm,shorten <=1mm,->](s4.east)->(s.west)node [midway,sloped,above,align=center] {};
		\draw [thick,shorten >=1mm,shorten <=1mm,->](s)|-(s5.west)node [midway,sloped,above,align=center] {};
		\draw [thick,shorten >=1mm,shorten <=1mm,->](s)|-(s6.west)node [midway,sloped,above,align=center] {};
		\node (s7) at (16,1) [] {\Large$\hat{S}$};
		\node (s8) at (15.4,-1) [] {\Large$\hat{X}$};
		\draw [thick,shorten >=1mm,shorten <=1mm,->](s5.east)->(s7.west)node [midway,sloped,above,align=center] {};
		\draw [thick,shorten >=1mm,shorten <=1mm,->](s6.east)->(s8.west)node [midway,sloped,above,align=center] {};
		
		\node (dx) at (7,-2) [] {\Large$d_X(x,\hat{x})\leq D_x$};
		\draw [dashed, -latex] (3.7,-0.2) |- (dx.west);
		\draw [dashed, -latex] (s8) |- (dx.east);
		\node (ds) at (7,-3) [] {\Large$d_S(s,\hat{s})\leq D_s$};
		\draw [dashed, -latex] (s1) |- (ds.west);
		\draw [dashed, -latex] (s7) |- (ds.east);
	\end{tikzpicture}
	